\documentclass[aps, prl, 11pt]{revtex4-1}
\usepackage{color}
\usepackage{latexsym}
\usepackage{amsmath}
\usepackage{amsfonts}
\usepackage{amsthm}
\usepackage{graphicx}
\usepackage{esint}
\usepackage{subfigure}
\usepackage{verbatim}
\usepackage{braket}
\usepackage{footmisc}
\usepackage{tikz}
\usetikzlibrary{shapes,arrows,calc}

\usepackage[bookmarks, hyperindex, colorlinks=true,
            linkcolor=red,
            urlcolor=blue, citecolor=blue]{hyperref}

 \definecolor{BLACK}{gray}{0}
 \definecolor{WHITE}{gray}{1}
 \definecolor{RED}{rgb}{1,0,0}
 \definecolor{GREEN}{rgb}{0,1,0}
 \definecolor{BLUE}{rgb}{0,0,1}
 \definecolor{CYAN}{cmyk}{1,0,0,0}
 \definecolor{MAGENTA}{cmyk}{0,1,0,0}
 \definecolor{YELLOW}{cmyk}{0,0,1,0}

\makeatother

\newcommand{\SAVE}[1]{}

\newcommand{\subfigimg}[3][,]{%
  \setbox1=\hbox{\includegraphics[#1]{#3}}
  \leavevmode\rlap{\usebox1}
  \rlap{\hspace*{42pt}\vspace*{12pt}\raisebox{\dimexpr\ht1-1.37\baselineskip}{#2}}
  \phantom{\usebox1}
}

\newcommand{\subfigimgone}[3][,]{%
  \setbox1=\hbox{\includegraphics[#1]{#3}}
  \leavevmode\rlap{\usebox1}
  \rlap{\hspace*{42pt}\vspace*{12pt}\raisebox{\dimexpr\ht1-1.37\baselineskip}{#2}}
  \phantom{\usebox1}
}

\newcommand{\subfigimgtwo}[3][,]{%
  \setbox1=\hbox{\includegraphics[#1]{#3}}
  \leavevmode\rlap{\usebox1}
  \rlap{\hspace*{42pt}\vspace*{12pt}\raisebox{\dimexpr\ht1-1.37\baselineskip}{#2}}
  \phantom{\usebox1}
}

\graphicspath{{./Figures/}}
\begin{document}

\title[Density matrix downfolding]{From real materials to model Hamiltonians with density matrix downfolding}
\author{Huihuo Zheng\,$^{1}$, Hitesh J. Changlani\,$^{2}$, \\Kiel T. Williams\,$^{3}$, Brian Busemeyer\,$^{3}$, and Lucas K. Wagner\,$^{3}$} 
\email{lkwagner@illinois.edu}
\affiliation{$^{1}$Argonne Leadership Computing Facility, Argonne National Laboratory, Lemont, IL, USA \\
$^{2}$Department of Physics and Astronomy and Institute for Quantum Matter, Johns Hopkins University, Baltimore, MD, USA\\
$^{3}$Department of Physics and Institute for Condensed Matter Theory, University of Illinois at Urbana-Champaign, Urbana, IL, USA} 
\date{\today}
\begin{abstract}
Due to advances in computer hardware and new algorithms, it is now possible to perform highly accurate many-body simulations of realistic materials with all their intrinsic complications.
The success of these simulations leaves us with a conundrum: how do we extract useful physical models and insight from these simulations? 
In this article, we present a formal theory of downfolding--extracting an effective Hamiltonian from first-principles calculations. 
The theory maps the downfolding problem into fitting information derived from wave functions sampled from a low-energy subspace of the full Hilbert space. 
Since this fitting process most commonly uses reduced density matrices, we term it density matrix downfolding (DMD).
\end{abstract}
\maketitle
\section{Introduction to downfolding the many electron problem}

In multiscale modeling of many-particle systems, the effective Hamiltonian (or Lagrangian) is one of the most core concepts. 
The effective Hamiltonian dictates the behavior of the system on a coarse-grained level, where `sub-grid' effects are folded into the parameters and form of the effective Hamiltonian. 
Many concepts in condensed matter physics can be viewed as statements about the behavior of the effective Hamiltonian. 
In particular, identification of `strongly correlated' materials as materials where band theory is not an accurate representation of the systems is a statement about effective Hamiltonians.
Effective Hamiltonians at different length scales also form the basis of the renormalization group~\cite{Wilson}.
A major goal in condensed matter physics is to determine what effective Hamiltonians apply to different physical situations, in particular quantum effective Hamiltonians, which lead to large-scale emergent quantum phenomena. 

The dominant effective model for quantum particles in condensed systems is band structure, and for metals, Fermi liquid theory. 
However, a major challenge is how this paradigm should be altered when it is no longer a good description of the physical system.
Examples of these include the high-T$_c$ cuprates and other transition metal oxides, which do not appear to be well-described by these simple effective Hamiltonians. 
For these systems, many models have been proposed, such as the Hubbard~\cite{Hubbard1963}, Kanamori~\cite{Kanamori1963}, $t$-$J$~\cite{tJSpalek} and Heisenberg models.
While these models have been extensively studied analytically and numerically, and have significantly enhanced our understanding of the physics of correlated electrons, their effectiveness for describing a real complex system of interest is often unclear. 
At the same time, more complex effective models can be commensurately more difficult to solve, so one would like to also find an accurate effective model that is computationally tractable.

To address the need for a link between {\it ab initio} electron-level models and larger scale models, downfolding has most commonly been carried out using approaches based on density functional theory (DFT). 
The one particle part is obtained from a standard DFT calculation which is projected onto localized Wannier functions and gives an estimate of the effective hoppings of the lattice model based on Kohn-Sham band structure calculations~\cite{Pavirini}. 
Then, to estimate the interactions, one assumes a model of screening of the Coulomb interactions based on constrained DFT, RPA, or some other methods. 
Since effects of interactions between the orbitals of interest have already been accounted for by DFT, a double counting correction is required to obtain the final downfolded Hamiltonian. 
The approach has been developed and widely applied~\cite{Pavirini, Dasgupta, Aryasetiawan2004, Jeschke2013}; but remains an active area of research~\cite{Haule_doublecounting}.
There are other downfolding approaches that include the traditional L\"owdin method, coupled to a stochastic approach~\cite{Tenno,Zhou_Ceperley} and the related method of canonical transformations~\cite{White_CT, Yanai_CT}. 
While they have many advantages, it is typically not possible to know if a given model {\it ansatz} was a good guess or not, and it is very rare for a technique to provide an estimate of the quality of the resultant model. 

The situation described above stands in contrast to the derivation of effective classical models. 
For concreteness, let us discuss classical force fields computed from {\it ab initio} electronic structure calculations. 
Typically, a data set is generated using an {\it ab initio} calculation in which the positions of the atoms and molecules are varied, creating a set of positions and energies. 
The parameters in the force field {\it ansatz} are varied to obtain a best-fit classical model.
Then, using standard statistical tools, it is possible to assess how well the fit reproduces the {\it ab initio} data within the calculation, without appealing to experiment. 
While translating that error to error in properties is not a trivial task, this approach has the important advantage that in the limit of a high quality fit and high quality {\it ab initio} results, the resultant model is predictive.

Na\"ively, one might think to reconcile the fitting approach used in classical force fields with quantum models by matching eigenstates between a quantum model and {\it ab initio} systems, varying the model parameters until the eigenstates match~\cite{Wagner2013}. 
However, this strategy does not work well in practice because it is often not possible to obtain exact eigenstates for either the model or the {\it ab initio} system.
To resolve this, we develop a general theory for generating effective quantum models that is exact when the wave functions are sampled from the manifold of low-energy states. 
Because this method is based on fitting the energy functional, we will show the practical application of this theory using both exact solutions and {\it ab initio} quantum Monte Carlo (QMC) to derive several different quantum models.

The endeavor we pursue here is to develop a multi-scale approach in which the effective interactions between quasiparticles (such as dressed electrons) are determined after an \textit{ab initio} simulation (but not necessarily exact solution) of the continuum Schroedinger equation involving all the electrons. 
The method uses reduced density matrices (RDMs), of low-energy states, not necessarily eigenstates, 
to cast downfolding as a fitting problem.  
We thus call it density matrix downfolding (DMD). 
In this paper, our application of DMD to physical problems employ one body (1-RDM) and two body (2-RDM) density matrices. 
The many-body states used in DMD will typically be generated using QMC techniques [either variational Monte Carlo (VMC) or diffusion Monte Carlo (DMC)] to come close to the low energy manifold.

The remainder of the paper is organized as follows:
\begin{itemize} 
\item In Section \ref{sec:theory}, we clarify and make precise what it means to downfold 
a many-electron problem to a few-electron problem. We recast the problem into minimization 
of a cost function that needs to be optimized to connect the many and few body problems. We further 
these notions both in terms of physical as well as information science descriptions, which allows us to connect to compression algorithms in the machine learning literature. 
\item Section \ref{sec:examples} discusses several representative examples where we consider multiband lattice models 
and {\it ab initio} systems to downfold to a simpler lattice model. 
\item In Section \ref{sec:conclusion}, we discuss future prospects of applications of the DMD method, ongoing challenges 
and clear avenues for methodological improvements. 
\end{itemize}

\newtheorem{theorem}{Theorem}
\newtheorem{definition}{Definition}
\newtheorem{lemma}{Lemma}

\section{Downfolding as a compression of the energy functional}
\label{sec:theory}
\subsection{Theory} 

\subsubsection{Energy functional}
Suppose we start with a quantum system with Hamiltonian $H$ and Hilbert space ${\mathcal H}$.

\begin{definition}
Let the energy functional be $E[\Psi] = \frac{\bra{\Psi}H\ket{\Psi}}{\braket{\Psi|\Psi}}$ for a wavefunction $\ket{\Psi} \in {\mathcal H}$.
\end{definition}

\begin{theorem}
\label{theorem:criticalpoint}
$E[\Psi]$ has a critical point only where $\Psi$ is an eigenstate of $H$.
\end{theorem}
\begin{proof}
\begin{eqnarray}
\frac{\delta }{\delta \Psi^*}  E[\Psi] = \frac{\delta}{\delta \Psi^*}\frac{\langle \Psi |H|\Psi\rangle}{\langle \Psi |\Psi\rangle} = \frac{H|\Psi\rangle}{\langle \Psi |\Psi\rangle} - \langle \Psi |H|\Psi\rangle \frac{|\Psi \rangle}{|\langle \Psi | \Psi\rangle|^2} =\frac{ (H-E[\Psi])|\Psi\rangle }{\langle\Psi|\Psi\rangle}\,.
\end{eqnarray}
Therefore, 
$\frac{\delta }{\delta \Psi^*}  E[\Psi] = 0$ if and only if $(H-E[\Psi])|\Psi\rangle =0$, i.e., $\Psi$ is an eigenvector of $H$ corresponding to eigenvalue $E[\Psi]$.  
\end{proof}

\subsubsection{Low energy space} 

\begin{definition}
Let $\mathcal{LE}(H,N)$ be a subset of ${\mathcal H}$ spanned by $N$ vectors given by the lowest energy solutions to $H\ket{\Psi_i}=E_i{\Psi_i}$. 
\end{definition}

\begin{definition}
$H_{eff}$ is an operator on the Hilbert space ${\mathcal {LE}(H,N)}$.	 
\end{definition}

\begin{definition}
The effective model $E_{eff}[\Psi]=\frac{\bra{\Psi}H_{eff}\ket{\Psi}}{\braket{\Psi|\Psi}}$ is a functional from $\mathcal{LE} \rightarrow \mathbb{R}$. 
\end{definition}

If $\ket{\Psi}\in \mathcal{LE}$ and $\ket{\Phi}\in {\mathcal H} \setminus \mathcal{LE}$, then $\ket{\Psi} \oplus \ket{\Phi} \in {\mathcal H}$.
In the following, we will use the direct sum operator $\oplus 0$ to translate between the larger ${\mathcal H}$ and the smaller $\mathcal{LE}$. 

\begin{lemma}
\label{lemma:zeroderiv}
Suppose that $\ket{\Psi}\in \mathcal{LE}$ and $\ket{\Phi} \in {\mathcal H} \setminus \mathcal{LE}$. 
Then $\left.\frac{\delta E[\Psi \oplus \Phi]}{\delta \Phi}\right|_{\Phi=0}=0$. 
\end{lemma}
\begin{proof}
$\langle \Psi\oplus 0 | H | 0\oplus \Phi \rangle=0$ because the two states have non-overlapping expansions in the eigenstates of $H$. 
Using that fact, we can evaluate
\begin{align}
\left.\frac{\delta E[\Psi \oplus \Phi]}{\delta \Phi}\right|_{\Phi=0} &= \left.\frac{\left(H-E[\Psi\oplus\Phi]\right) \ket{\Phi} }{\braket{\Psi|\Psi} + \braket{\Phi|\Phi}} \right|_{\Phi=0} = 0.
\end{align}
\end{proof}
This is equivalent to noting that $H$ is block diagonal in the partitioning of ${\mathcal H}$ into $\mathcal{LE}$ and ${\mathcal H} \setminus \mathcal{LE}$.
Importantly, if $\ket{\Psi} \in \mathcal{LE}$, then $\frac{\delta  E[\Psi\oplus 0] }{\delta (\Psi\oplus 0)^*} = \ket{\Psi'} \oplus 0$, where $\ket{\Psi'} \in \mathcal{LE} $. 

\begin{theorem}
\label{theorem:matching}
Assume $ E[\Psi\oplus 0]  = E_{eff}[\Psi]+C$ for any $\ket{\Psi } \in \mathcal{LE}$, where $C$ is a constant. 
Then $(H_{eff}+C)|\Psi\rangle\oplus 0 = H (|\Psi\rangle \oplus 0)$.
\end{theorem}
\begin{proof}
Note that
\begin{align}
	\frac{\delta E[\Psi\oplus 0]}{\delta (\Psi\oplus 0)^*}=\frac{(H-E[\Psi\oplus 0])\ket{\Psi\oplus 0} }{\braket{\Psi|\Psi}}
	\label{eqn:psider}
\end{align}
and 
\begin{align}
	\frac{\delta E_{eff}[\Psi]}{\delta \Psi^*}=\frac{(H_{eff}-E_{eff}[\Psi])\ket{\Psi} }{\braket{\Psi|\Psi}}.
	\label{eqn:psieffder}
\end{align}
Since the derivatives are equal, setting Eq.~\eqref{eqn:psider} equal to Eq.~\eqref{eqn:psieffder},
\begin{align}
	 H\ket{\Psi\oplus 0}= (H_{eff}+E[\Psi\oplus 0]-E_{eff}[\Psi])\ket{\Psi}\oplus 0 =(H_{eff}+C)\ket{\Psi}\oplus 0.
\end{align}
\end{proof}

Theorem~\ref{theorem:matching} combined with Lemma~\ref{lemma:zeroderiv} means that the eigenstates of $H_{eff}$ are be the same as the eigenstates of $H$ if its derivatives match $H$. 
Such $H_{eff}$ always exists. 
Let $H_{eff} = \sum_{i}^N E_i |\Psi_i\rangle \langle \Psi_i|$ where $|\Psi_i\rangle$'s are eigenstates belong to $\mathcal{LE}(H,N)$. This satisfies $E[\Psi] = E_{eff}[\Psi]$ and $H_{eff}|\Psi\rangle = H |\Psi \rangle$ for any $\Psi$ in $\mathcal{LE}(H,N)$.  

We have thus reduced the problem of finding an effective Hamiltonian $H_{eff}$ that reproduces the low-energy spectrum of $H$ to matching the corresponding energy functionals $E[\Psi]$ and $E_{eff}[\Psi]$. 
This involves sampling the low-energy space, choosing the form of $H_{eff}$, and optimizing the parameters.
An important implication of this is that it is not necessary to diagonalize either of the Hamiltonians; one must only be able to select wave functions from the low-energy space $\mathcal{LE}$.
As we shall see, this can be substantially easier than attaining eigenstates.

Some further notes about this derivation:
\begin{itemize}
\item Fitting $\Psi$'s must come from $\mathcal{LE}$. It is not enough that the energy functional $E[\Psi]$ is less than some cutoff.
\item In the case of sampling an approximate $\mathcal{LE}$, the error comes from non-parallelity of $E[\Psi]$ with the correct low energy manifold, up to a constant offset.
\item While $H_{eff}$ is unique, it has many potential representations and approximations. 
\item Our method can be applied to any manifold spanned by eigenstates
\item Model fitting is finding a compact approximation to $E_{eff}[\Psi]$. This is a high-dimensional space, so we use descriptors to do this.	
\item For operators that are not the Hamiltonian, it is possible to fit $\mathcal{O}_{eff}[\Psi] \simeq {\mathcal O}[\Psi]$ in a similar way. However, the eigenstates of ${\mathcal O}$ and ${\mathcal O}_{eff}$ will not coincide in general unless $\mathcal{O}$ commutes with the Hamiltonian.
\end{itemize}

The theory presented above maps coarse-graining into a functional approximation problem. 
This is still rather intimidating, since even supposing one can generate wave functions in the low-energy space, they are still complicated objects in a very large space.
An effective way to accomplish this is through the use of descriptors, $d_i[\Psi]$, which map from ${\mathcal H} \rightarrow \mathbb{R}$.
Then we can approximate the energy functional as follows
\begin{equation}
E_{eff}[\Psi] \simeq \sum_i f_i(d_i[\Psi]),
\end{equation}
where $f_i$ are some parameterized functions.
This will allow us to use techniques from statistical learning to efficiently describe $E_{eff}$. 
\subsection{Practical protocol}

\tikzstyle{decision} = [diamond, draw, fill=blue!10, 
    text width=4.5em, text badly centered, node distance=3cm, inner sep=0pt]
\tikzstyle{block} = [rectangle, draw, fill=blue!10, 
    text width=5em, text centered, rounded corners, minimum height=4em]
\tikzstyle{result} = [rectangle, draw, fill=red!10, 
    text width=5em, text centered, rounded corners, minimum height=4em]
\tikzstyle{line} = [draw,-latex',very thick]
\tikzstyle{cloud} = [draw, ellipse,fill=red!20, node distance=3cm,
    minimum height=2em]
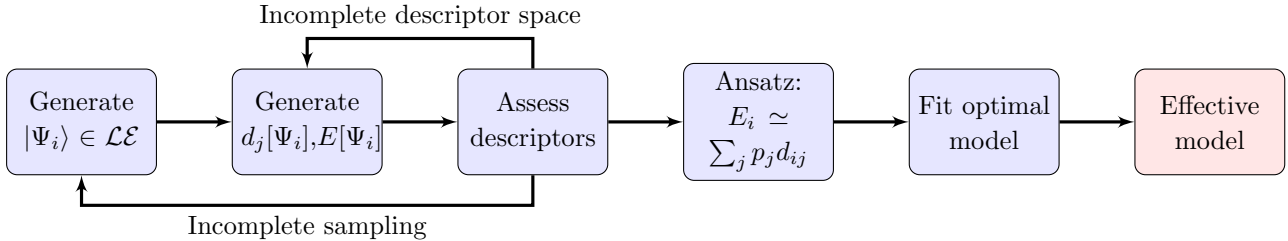
\begin{figure*}[hbt]
\begin{tikzpicture}[scale=2,node distance = 3cm, auto]
    \node [block] (wfs) {Generate $\ket{\Psi_i} \in \mathcal{LE}$};
    \node [block, right of=wfs] (descriptors) {Generate $d_j[\Psi_i]$,$E[\Psi_i]$};
    \node [block, right of=descriptors] (assess) {Assess descriptors};
    \node [block, right of=assess] (ansatz) {Ansatz: $E_i \simeq \sum_j p_j d_{ij}$};
    \node [block, right of=ansatz] (fit) {Fit optimal model};
    \node [result, right of=fit] (model) {Effective model};
    \path [line] (wfs) -- (descriptors);
    \path [line] (descriptors) -- (assess);
    \path [line] (assess) --  (ansatz);
    \path [line] (ansatz) --  (fit);
    \path [line] (fit) --  (model);

    \path [line] (assess.south) -- ($ (assess.south) + (0,-0.2)$) 
                 -- node [below] {Incomplete sampling} 
                 ($ (wfs.south) + (0,-0.2)$) --  (wfs.south);
    \path [line] (assess.north) -- ($ (assess.north) + (0,0.2)$) 
                 -- node [above] {Incomplete descriptor space} 
                 ($ (descriptors.north) + (0,0.2)$) --  (descriptors.north);

\end{tikzpicture}
\caption{A practical protocol for fitting effective models to {\it ab initio} data.}
\label{fig:protocol} 
\end{figure*}

A practical protocol is presented in Figure~\ref{fig:protocol}. 
In this section we go through this procedure step by step.

\paragraph{Generating $\ket{\Psi_i}\in \mathcal{LE}$}
Ideally one would be able to sample the entire low-energy space. 
Typically, however, the space will be too large and it will need to be sampled. 
The optimal wave functions to use depend on the models one expects to fit, which we will discuss in detail  in later steps. 
Simple strategies that we will use in the examples below include excitations with respect to a determinant and varying spin states.

\paragraph{Generate $d_j[\Psi_i]$ and $E[\Psi_i]$} 
The choice of descriptor is fundamental to the success of the downfolding. 
In the case of a second-quantized Hamiltonian
\begin{equation}
H_{eff} = E_0 + \sum_{ij} t_{ij} (c_i^\dagger c_j + h.c.) + \sum_{ijkl} V_{ijkl} c_i^\dagger c_j^\dagger c_k c_l,
\end{equation}
a set of linear descriptors by simply taking the expectation value of both sides of the equation. 
Then for example, the occupation descriptor for orbital $k$ is $d_{occ(k)}[\Psi_i] = \braket{\Psi_i | c_k^\dagger c_k | \Psi_i}$; the double occupation descriptor for orbital $k$ is $d_{double(k)}[\Psi_i] = \braket{\Psi_i | n_{k\uparrow}n_{k\downarrow} | \Psi_i}$. 
The orbital that $c_k$ represents is part of the descriptor, and in the examples below we will discuss this choice as well.
One is not limited to static orbital descriptors; they may have a more complex functional dependence on the trial function to include orbital relaxation.

\paragraph{Assess descriptors}
At this point, one has collected the data $E_i$ and $d_{ij}$. 
If two descriptors have a large correlation coefficient, then they are redundant in the data set. 
This could either mean that the sampling of the low-energy Hilbert space $\mathcal{LE}$ was insufficient, or that they are both proxies for the same differences in states. 
If two data points have the same or very similar descriptor sets, but different energies, then either the descriptor set is not enough to describe the variations in the low-energy space, or the sampling has generated states that are not in the low-energy space.
To resolve these possibilities, one should analyze the difference between the two wave functions.  

In either case, when the model is accurate, the fits will be accurate.
If descriptors values available in the reduced Hilbert space are not represented in the sampled wave functions, then intruder states can appear upon solution of the effective model. 
In that case, the model fitting is an extrapolation instead of an interpolation.
For this reason it is desirable to have eigenstates or near-eigenstates in the sample set if possible; they are guaranteed to be on the corners of the descriptor space if the model is accurate.

\paragraph{Ansatz: $E_i \simeq \sum_i  d_{ij} p_j$} 
If the descriptors are chosen well, then the model can be written in linear form:
\begin{equation}
E[\Psi_i] = \sum_j p_j d_j[\Psi_i],	
\end{equation}
which we shorten to 
\begin{equation}
{\bf E} = D{\bf p} .
\label{eqn:EdP}
\end{equation}
If this can be done, the fitting problem is reduced to a linear regression optimization.
More complex functions of the descriptors are also possible, although at the cost of making the effective model more difficult to solve and complicating the fitting procedure.

\paragraph{Fit optimal model}
Finally, one wishes to find a set of parameters such that Eq.~\eqref{eqn:EdP} is satisfied as closely as possible. 
There are many choices to make in this step, which will often depend on the desired properties of the final model. 
One can imagine choosing different cost functions to minimize, which can also include a penalty for complicated models. 
In our tests, we have successfully used LASSO \cite{Lasso} and matching pursuit techniques \cite{MP_Zhang1993} to select high quality and compact model parameters. 
A detailed example of using the latter technique is presented in Section~\ref{subsection:fese}.

\section{Representative Examples}
\label{sec:examples}
Given the theoretical framework for downfolding a many orbital (or many-electron) problem to a 
few orbital (or few-electron) problem, we now discuss examples which elucidate the DMD method. 
The examples are as follows:
\begin{itemize}
\item Section~\ref{subsection:3band}: Three-band Hubbard $\rightarrow$ one-band Hubbard at half filling. Demonstrates finding a basis set for the second quantized operators and uses a set of eigenstates directly sampled from the low-energy space to find a one-band model.
\item Section~\ref{subsection:1dhydrogen}: Hydrogen chain $\rightarrow$ one-band Hubbard model at half filling. Demonstrates basis sets for {\it ab initio} systems and the possibility to use this technique to determine the quality of a model to a given physical situation.
\item Section~\ref{subsection:graphene}: Graphene $\rightarrow$ one-band Hubbard model with and without $\sigma$ electrons. Demonstrates using the downfolding procedure to examine the effects of screening due to core electrons. 
\item Section~\ref{subsection:fese}: FeSe molecule $\rightarrow$ $3d,4p,4s$ system. Demonstrates the use of matching pursuit to assess the importance of terms in an effective model and to select compact effective models.
\end{itemize}

In all examples we will highlight the important ingredients associated with DMD. 
First and foremost is the choice of low energy space or energy window i.e. how our database of wave functions was generated. 
Associated with this is the choice of the one body space in terms of which the effective Hamiltonian is expressed. 
Finally, we discuss aspects of the functional forms or parameterizations that are expected to describe our physical problem. 
An important effective Hamiltonian that enters three out of our four representative examples is the one-band or single-band Hubbard model:
\begin{equation}
	H = E_0 -t \;\sum_{\langle i,j \rangle, \eta} \tilde{d}_{i,\eta}^{\dagger} \tilde{d}_{j,\eta} + U \;\sum_{i} \tilde{n}^{i}_{\uparrow} \tilde{n}^{i}_{\downarrow}\,,
\label{eq:oneband}
\end{equation}
where $t$ and $U$ are downfolded (renormalized) parameters, $\eta$ is a spin index, 
$\tilde{d}_{i,\eta}$ is the effective one-particle operator associated with spatial orbital (or site) $i$ 
and $n_{i,\eta}=\tilde{d}_{i,\eta}^{\dagger} \tilde{d}_{i,\eta}$ is the corresponding number operator.
$\langle i,j \rangle$ is used to denote nearest neighbor pairs.
We will sometimes drop the constant energy shift $E_0$ when we write equations like Eq.~\eqref{eq:oneband}.

\subsection{Three-band Hubbard model to one-band Hubbard model at half filling}
\label{subsection:3band} 
Our first example is motivated by the high $T_c$ superconducting cuprates~\cite{Bednorz1986} that 
have parent Mott insulators with rich phase diagrams on electron or hole doping~\cite{Dagotto_RevModPhys, LeeWen_RevModPhys}. 
Many works have been devoted to their model Hamiltonians and corresponding parameter 
values~\cite{tJSpalek, Pavirini, Emery, ZhangRice, Hybertsen_PRB1989, Hybertsen_PRB1990, Kent_Hubbard}. 
A minimal model involving both the copper and oxygen degrees of freedom 
is the three-orbital or three-band Hubbard model, 
\begin{eqnarray}
H &=&    \epsilon_p \sum_{j\in p,\eta} n_{j,\eta} + \epsilon_{d} \sum_{i \in d,\eta}  n_{i,\eta} 
        + t_{pd} \sum_{\langle i\in d ,j \in p \rangle, \eta} \text{sgn}(p_i,d_j) \Big( c_{i,\eta}^{\dagger} c_{j,\eta} + \text{h.c.} \Big) \nonumber \\
          & &   + U_p \sum_{j\in p} n_{j,\uparrow} n_{j,\downarrow} + U_d \sum_{i\in d} n_{i,\uparrow} n_{i,\downarrow} + V_{pd} \sum_{\langle i \in p ,j \in d \rangle} n_j n_i\,,
\end{eqnarray}
where $d_i,p_j$ refer to the  $d_{x^2 - y^2}$ orbitals of copper at site $i$ and $p_x$ or $p_y$ 
oxygen at site $j$, respectively. 
$\text{sgn}(p_i,d_j)$ is the sign of the hopping $t_{pd}$ 
between nearest neighbors, shown schematically in Figure~\ref{fig:threeband}. 
$\epsilon_d$ and $\epsilon_p$ are orbital energies, $U_d$ and $U_p$ are strengths of onsite Hubbard interactions,  
and $V_{pd}$ is the strength of the density-density interactions between a neighboring $p$ and $d$ orbital. 
To simplify we consider only the case where $\epsilon_p$, $U_d$ and $t_{pd}$ are non zero; $t_{pd}$ is chosen throughout this section to be the typical value of $1.3$ eV to give the reader a sense of overall energy scales. 
Since we work with fixed number of particles we set our reference zero energy to be $\epsilon_d = 0$, thus the charge transfer energy $\Delta \equiv \epsilon_p - \epsilon_d$ equals $\epsilon_p$ in our notation. 
We work in the hole notation; half filling corresponds to two spin-up and two spin-down holes on the $2\times2$ cell.

\begin{figure}[hbt]
\centering
\includegraphics[width=0.9\linewidth]{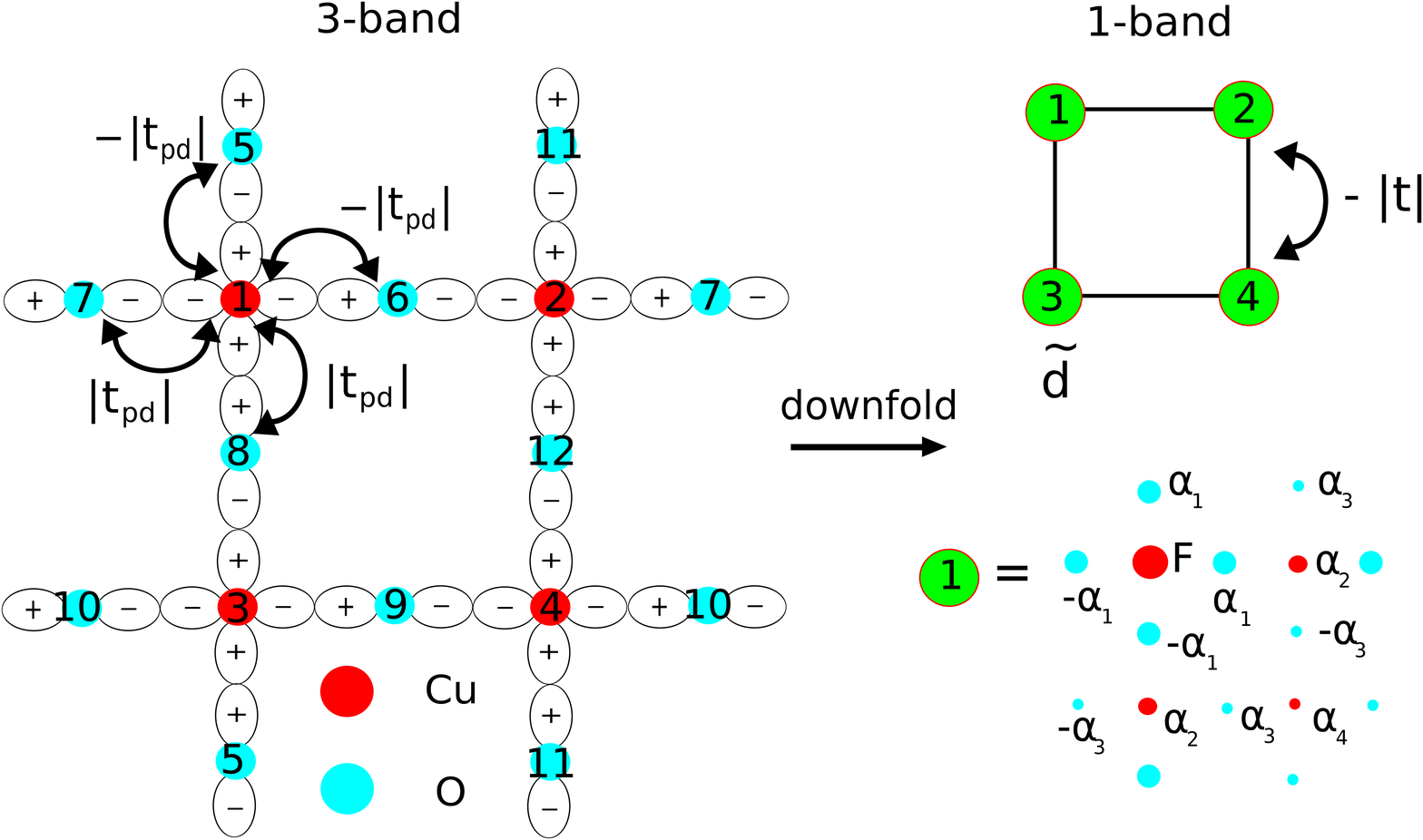}
\caption
{Schematic for downfolding the three-band Hubbard model on the $2\times2$ cell 
to the one-band Hubbard model. The oxygen orbitals are eliminated to give ``dressed" $d$-like orbitals of the one-band model, 
with modified hopping and interaction parameters. The relationship between the $\tilde{d}$ and the copper and oxygen orbitals is 
encoded by a linear transformation which is parameterized by $\alpha_1$, $\alpha_2$, $\alpha_3$, $\alpha_4$ and $F$ (see Appendix for more details). Here 
the diameter of the circles has been shown to be proportionate to the magnitude of $F$ or $\alpha_i$.}
\label{fig:threeband} 
\end{figure}

It is our objective to determine what one-band Hubbard model~[Eq.~\eqref{eq:oneband}]
``best" describes the three-band data. The effective \textit{d-like} orbitals $\tilde{d}_{i,\eta}$, 
that enter the low energy description are mixtures of copper and oxygen orbitals; this optimal transformation also remains an unknown. 
Thus the model determination involves two aspects (1) what are the composite objects that give a 
compact description of the low energy physics? and (2) given this choice what are the effective interactions between them? 
(A similar problem was posed and solved by one of us in the context of spin systems~\cite{Changlani_percolation}.)
In addition, the best effective Hamiltonian description depends on the energy scale of interest. 
All these issues will be addressed in the remainder of the section. 

We begin by encoding the relationship between the bare and effective operators as a linear transformation ${\bf T}$, 
\begin{equation}
	\tilde{d}_{i,\eta} = \sum_{j} T_{ij} c_{j,\eta}
\label{eq:dc}
\end{equation}
where $c_{j,\eta}$ is the hole (destruction) operator and refers to either the bare $d$ or $p$ orbitals. 
Further generalizations of this relationship (for example, including higher body terms) are also possible, but have not been considered here. 
For the $2\times2$ unit cell {\bf T} is a $4 \times 12 $ matrix, which we parameterize by 
four distinct parameters. These correspond to mixing of a copper orbital 
with nearest neighbor oxygens ($\alpha_1$), nearest neighbor coppers ($\alpha_2$), next-nearest neighbor oxygens ($\alpha_3$) 
and next-nearest neighbor coppers ($\alpha_4$) as shown schematically in Figure~\ref{fig:threeband}. 
The explicit form of ${\bf T}$ after accounting for the symmetries of the 
lattice has been written out in the Appendix. These parameters are optimized to minimize a certain cost function, 
which will be explained shortly. 

All RDMs in the three-band and one-band descriptions are also related via ${\bf T}$; 
the ones that we focus on are evaluated in eigenstate $s$ and are given by,
\begin{subequations}
\begin{eqnarray}
	\langle {\tilde{d}_{i,\eta}}^{\dagger} \tilde{d}_{j,\eta} \rangle_{s} &=& \sum_{mn} T^{*}_{im} \langle {c_{m,\eta}}^{\dagger} c_{n,\eta} \rangle_{s} T_{jn} \label{eq:dmstransformations1} \,,\\
	\langle \tilde{n}_{i,\uparrow} \tilde{n}_{i,\downarrow} \rangle_{s} &=& \sum_{jkmn} T^{*}_{ij} T^{*}_{im} \langle {c_{j,\uparrow}}^{\dagger} {c_{m,\downarrow}}^{\dagger} c_{n,\downarrow} c_{k,\uparrow} \rangle_{s} T_{in} T_{ik}\,.
\label{eq:dmstransformations2}
\end{eqnarray}
\end{subequations}
We optimize ${\bf T}$ by demanding two conditions be satisfied, (1) the effective orbitals ($\tilde{d}_{i,\eta}$) 
are orthogonal to each other i.e. $\Big({\bf T} {\bf T}^{\dagger}\Big)_{mn} = \delta_{mn}$
and (2) the sum of all diagonal entries (trace) of the 1-RDM of the effective orbitals for all low energy eigenstates 
equals the number of electrons of a given spin i.e. $\sum_{i} \sum_{\eta} \langle {\tilde{d}_{i,\eta}}^{\dagger} \tilde{d}_{i,\eta} \rangle_{s} = N_{\eta}$. 
These conditions are enforced by minimizing a cost function,
\begin{equation}
C = \sum_{s} \sum_{\eta} \Big( \sum_{i} \langle \tilde{d}_{i,\eta}^{\dagger} \tilde{d}_{i,\eta} \rangle_{s} - N_{\eta} \Big)^{2} + \sum_{mn} ( \Big({\bf T} {\bf T}^{\dagger}\Big)_{mn} -\delta_{mn})^{2}\,.
\label{eq:C}
\end{equation} 
For the $2\times2$ cell, $N_{\uparrow}=N_{\downarrow}=2$ and $i=1,2,3,4$. 
The number of states $s$ was varied from three to six, depending on the energy window of interest.  

Figure~\ref{fig:varyUdep} shows regimes of the three-band model where 
the lowest six eigenstates are separated from the higher energy manifold; the 
fourth and fifth eigenstates are degenerate. 
In the large $U_d$ limit, charge fluctuations are suppressed and these six 
states correspond to the Hilbert space of $4 \choose 2$ states of the effective spin model in its $S_z=0$ sector.
These states have primarily \textit{d-like} character, an aspect we will verify in this section. 
The eigenstates outside of this manifold involve \textit{p-like} excitations which the one-band model is not designed 
to capture. 

We chose the lowest three eigenstates of the three-band model for minimizing 
the cost in Eq.~\eqref{eq:C}. The four dimensional space of parameters of ${\bf T}$ 
was scanned for this purpose. The corresponding trace and orthogonality conditions are simultaneously 
satisfied with only small deviations, confirming the validity of Eq.~\eqref{eq:dc}. 
Importantly, the 1-RDM elements in the transformed basis corresponding to nearest neighbors $\langle \tilde{d}_1^{\dagger} \tilde{d}_2 \rangle_s$ 
already provide estimates for $U/t$ of the effective model. Since the exact knowledge of the corresponding eigenstates of 
the one-band Hubbard model is available for arbitrary $U/t$ by exact diagonalization, we directly look up the $U/t$ with 
the same 1-RDM value. These estimates complement the one obtained by DMD which was carried out 
with the same three low-energy eigenstates, using their energies and 
the computed values of $\langle \tilde{d}_1^{\dagger} \tilde{d}_2 \rangle_s$ 
and $\langle \tilde{n}_{i,\uparrow} \tilde{n}_{i,\downarrow} \rangle_{s}$ from Eqs.~\eqref{eq:dmstransformations1} 
and~\eqref{eq:dmstransformations2}.
~\footnote{We also mimicked the situation characteristic of \textit{ab initio} 
examples where no eigenstates are generally available. Several non eigenstates were generated as random linear combinations of 
the lowest three eigenstates and input into the DMD procedure, with 
similar outcomes.}
A representative example of our results for $U_{d}/t_{pd}=8$ and $\Delta/t_{pd}=3$ 
has been discussed in the Appendix. 
\renewcommand{\thesubfigure}{(\Alph{subfigure})}

\begin{figure}[hbt]
\centering
\subfigure[]{\includegraphics[width=0.47\linewidth]{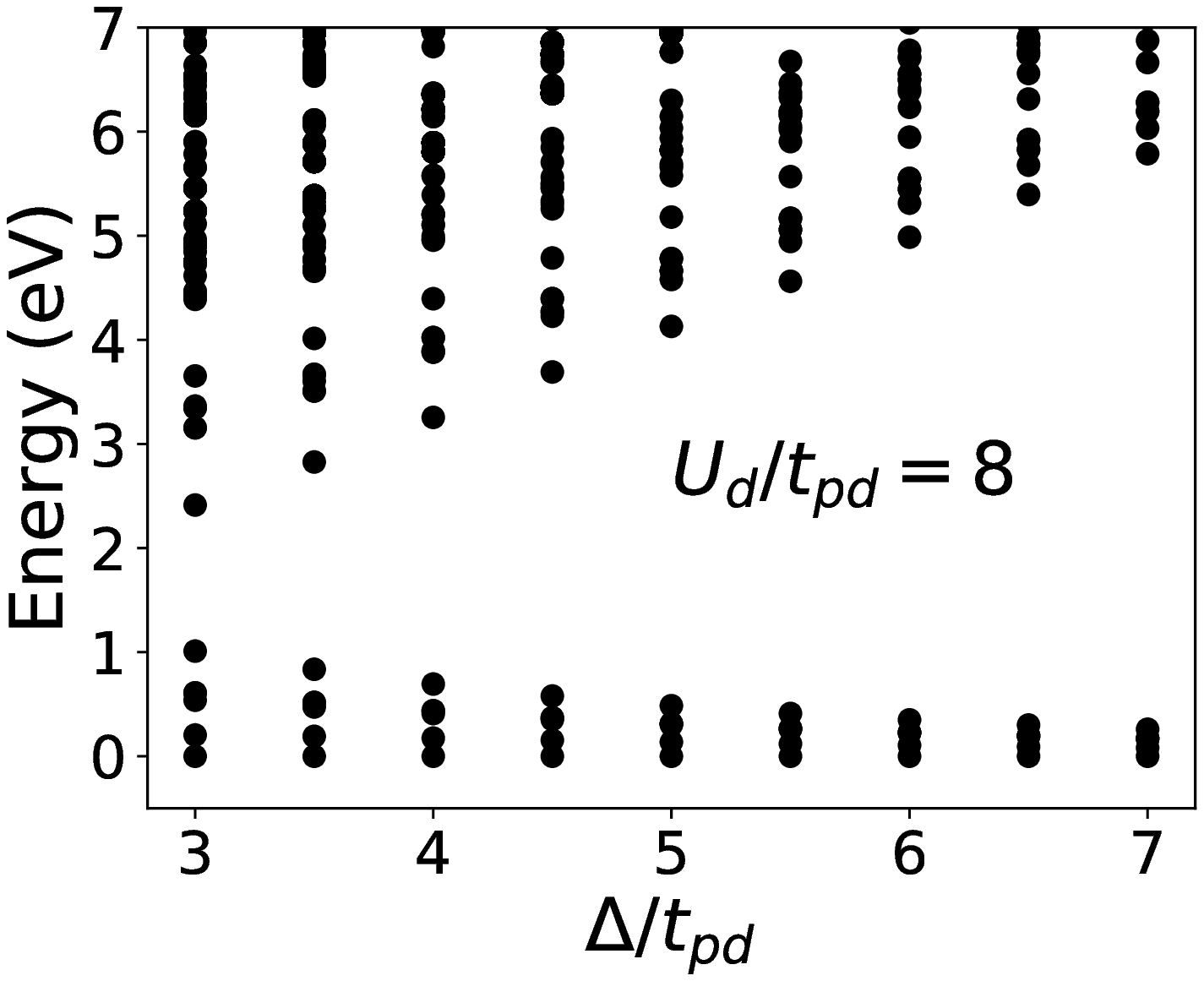}}
\subfigure[]{\includegraphics[width=0.52\linewidth]{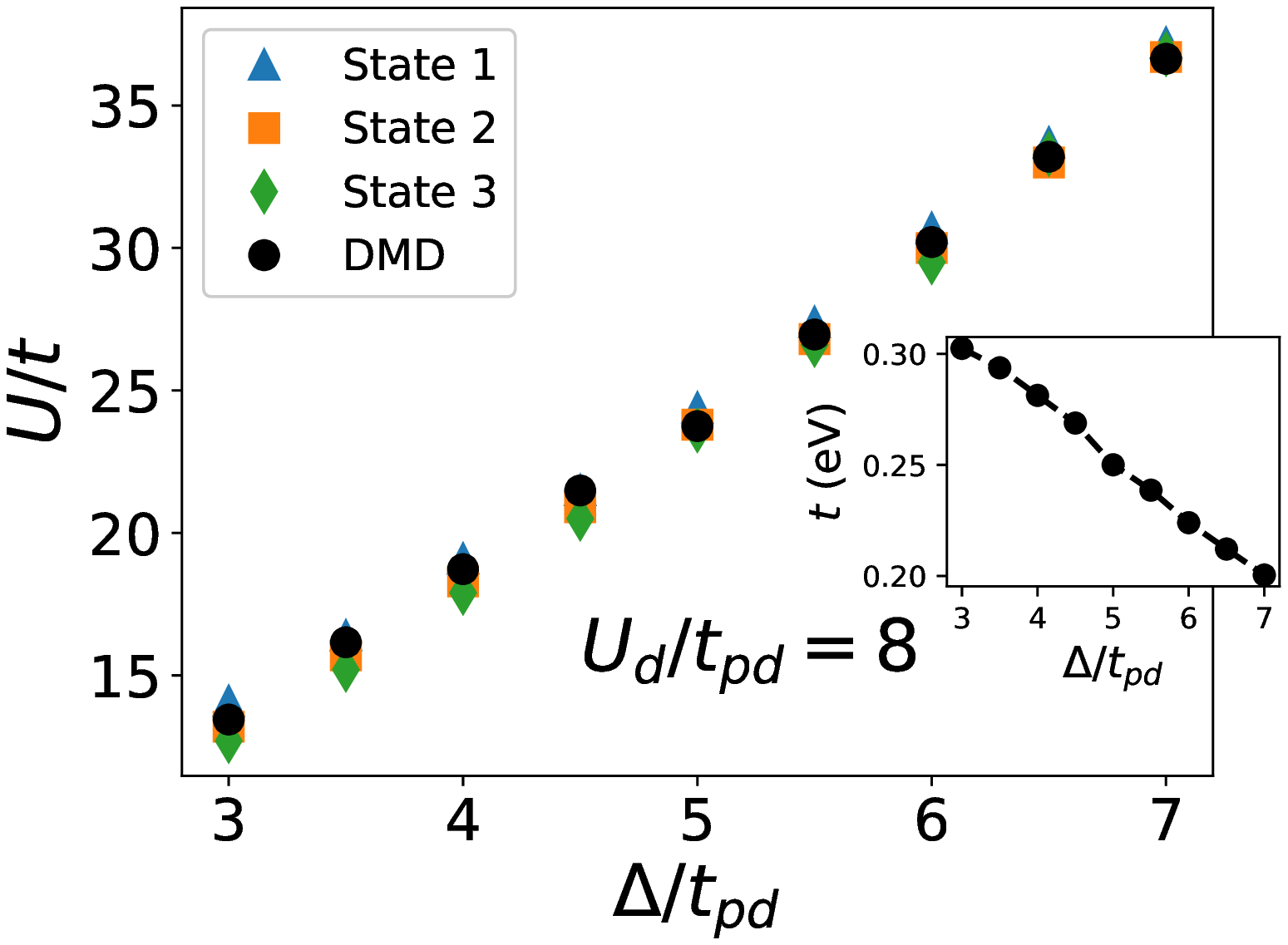}}
\subfigure[]{\includegraphics[width=0.47\linewidth]{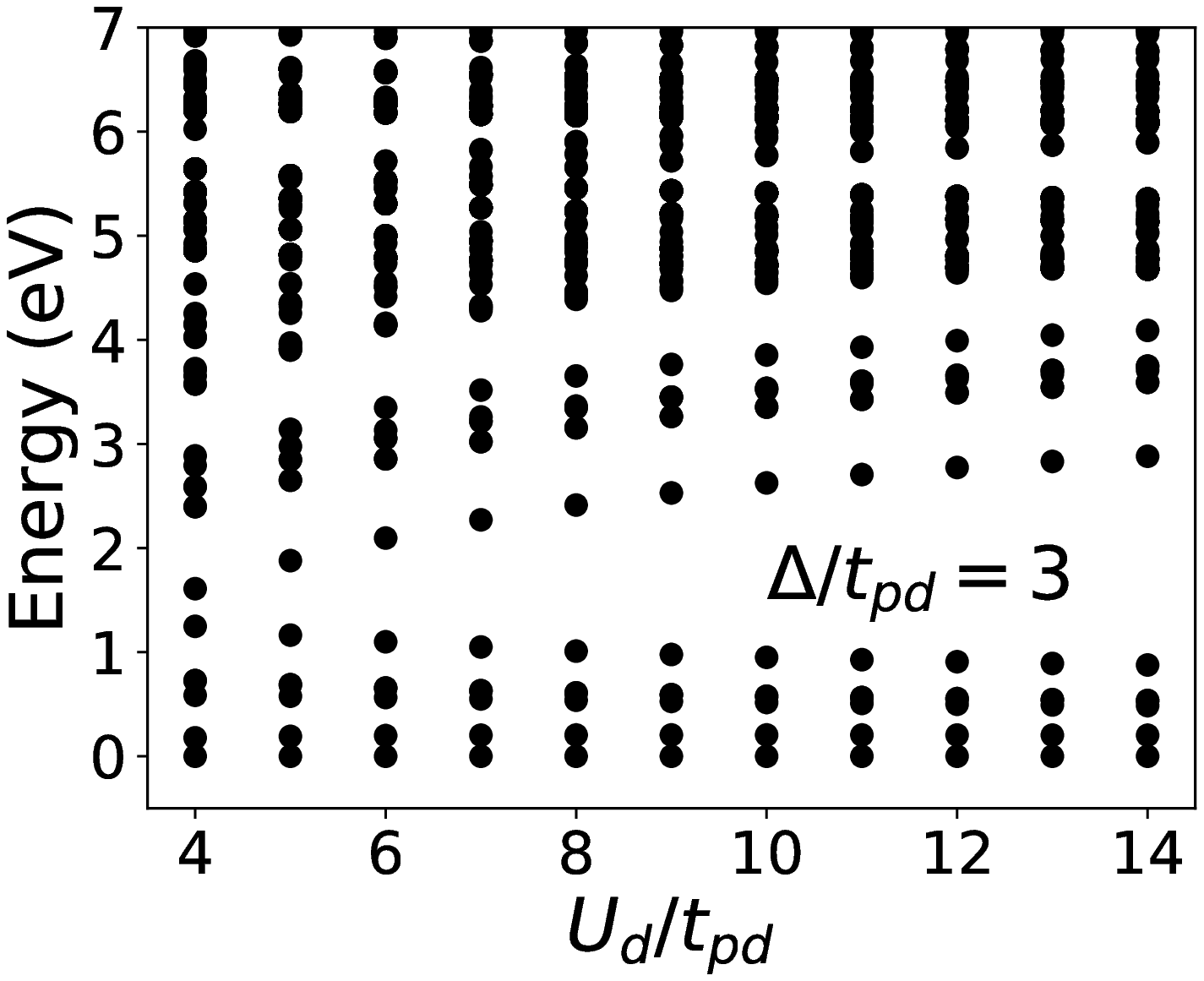}}
\subfigure[]{\includegraphics[width=0.52\linewidth]{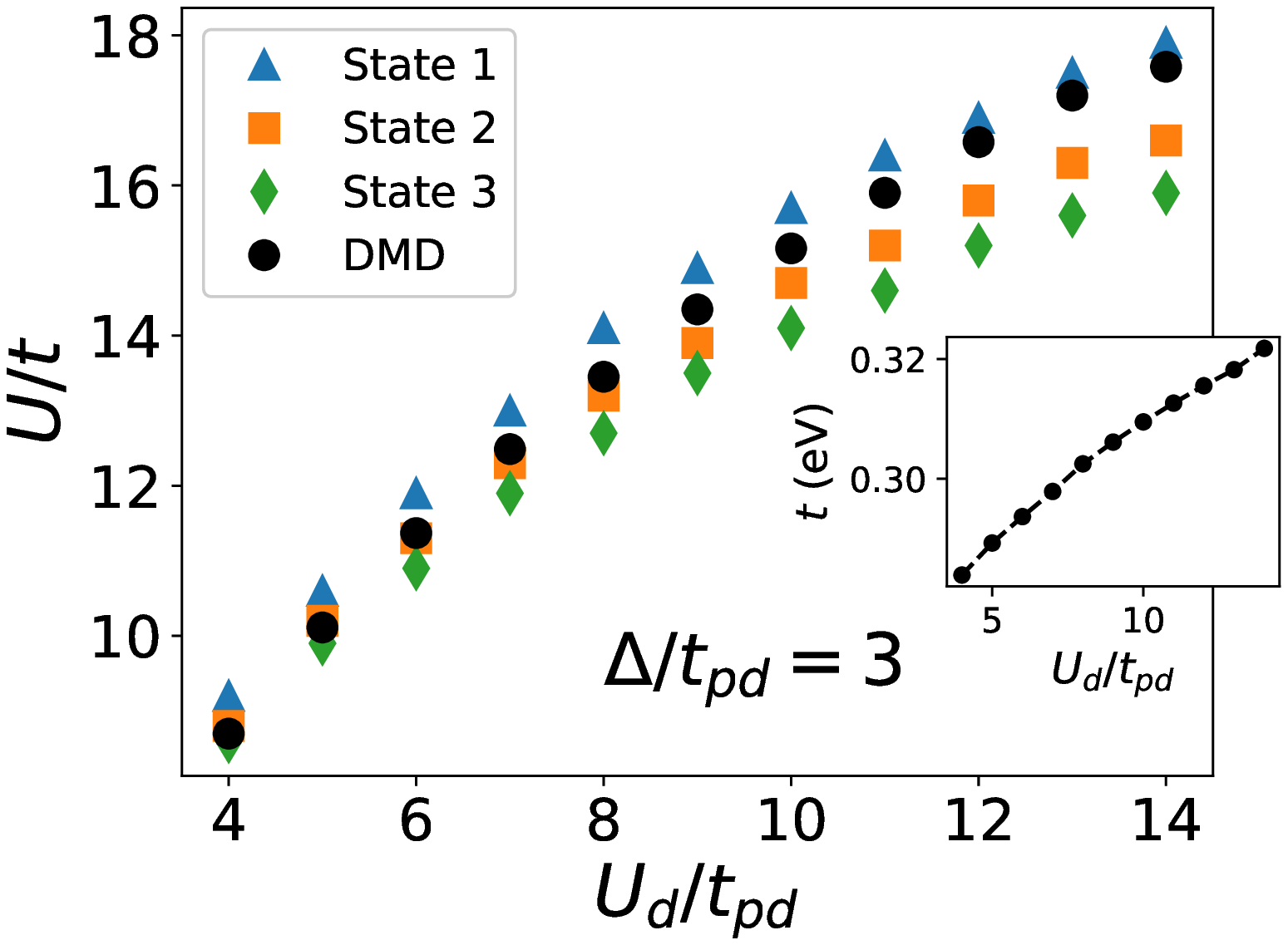}}
\caption{Downfolding of a three-band Hubbard model to an effective one-band Hubbard model of d-like orbitals. Panels (A) and (C) show the low energy spectra of the three-band model 
(relative to the corresponding ground state) for the cases of (A) fixed $U_d/t_{pd}=8$ and 
varying $\Delta/t_{pd}$ and (C) fixed $\Delta/t_{pd}=3$ and varying $U_d/t_{pd}$. 
In all cases, $t_{pd}$ is set to be $1.3$ eV. 
Panels (B) and (D) show the downfolded parameters for the one-band model corresponding to the three-band parameter 
choices in (A) and (C) respectively.  
The one-band $U/t$ values were obtained either by comparing $\langle \tilde{d}_1^{\dagger} \tilde{d}_2\rangle_s$ 
with the corresponding one-band model eigenstates or by the DMD procedure using 
the lowest three eigenstates. The insets show $t$ obtained from DMD. }
\label{fig:varyUdep} 
\end{figure}

Some trends in the one-band description are explored in Figure~\ref{fig:varyUdep} 
by monitoring the downfolded parameters as a function of varying $\Delta/t_{pd}$ and $U_d/t_{pd}$. 
For example, when $U_d/t_{pd}=8$ is fixed and $\Delta/t_{pd}$ is increased, we find that 
the effective hopping $t$ decreases and $U/t$ increases. This is physically reasonable since an increasing difference in the 
single particle energies of the copper and oxygen orbitals makes it energetically unfavorable for holes 
to hop between the two orbitals. When $\Delta/t_{pd}=3$ is fixed and $U_d/t_{pd}$ is increased, $U/t$ increases. 
As one mechanism of avoiding the large $U_d$, the copper orbitals are forced to hybridize more with the oxygen ones; 
on the other hand, hole delocalization is suppressed in a bid to maintain mostly one hole per $\tilde{d}$ due to the larger 
$U/t$. The net result of these effects is that the $t$ also increases.

An important check for the one-band model is its ability to reproduce the low energy gaps of the three-band model; these have been compared in Figure~\ref{fig:energyfit}. 
For the case of $\Delta/t_{pd}=3$, we observe that for all $U_d/t_{pd}$ the lowest three eigenstates were reproduced well. 
This model also reproduces the states outside of the DMD energy window, although with slightly larger errors. 
Similar trends are seen for the case of $\Delta/t_{pd}=5$, 
with the noticeable difference being that the energy error of the highest state has reduced. 
This also reflects that the parameters obtained from DMD are, in general, dependent on the energy window of interest, a 
point which we will highlight shortly by investigating it systematically. 

\renewcommand{\subfigimgone}[3][,]{%
  \setbox1=\hbox{\includegraphics[#1]{#3}}
  \leavevmode\rlap{\usebox1}
  \rlap{\hspace*{110pt}\vspace*{1200pt}\raisebox{\dimexpr\ht1-1.7\baselineskip}{#2}}
  \phantom{\usebox1}
}
\renewcommand{\subfigimgtwo}[3][,]{%
  \setbox1=\hbox{\includegraphics[#1]{#3}}
  \leavevmode\rlap{\usebox1}
  \rlap{\hspace*{110pt}\vspace*{1200pt}\raisebox{\dimexpr\ht1-1.9\baselineskip}{#2}}
  \phantom{\usebox1}
}
\begin{figure}[tbh]
\centering
 \begin{tabular}{@{}p{0.90\linewidth}@{\quad}p{\linewidth}@{}}
\subfigimgone[width=0.49\linewidth]{(A)}{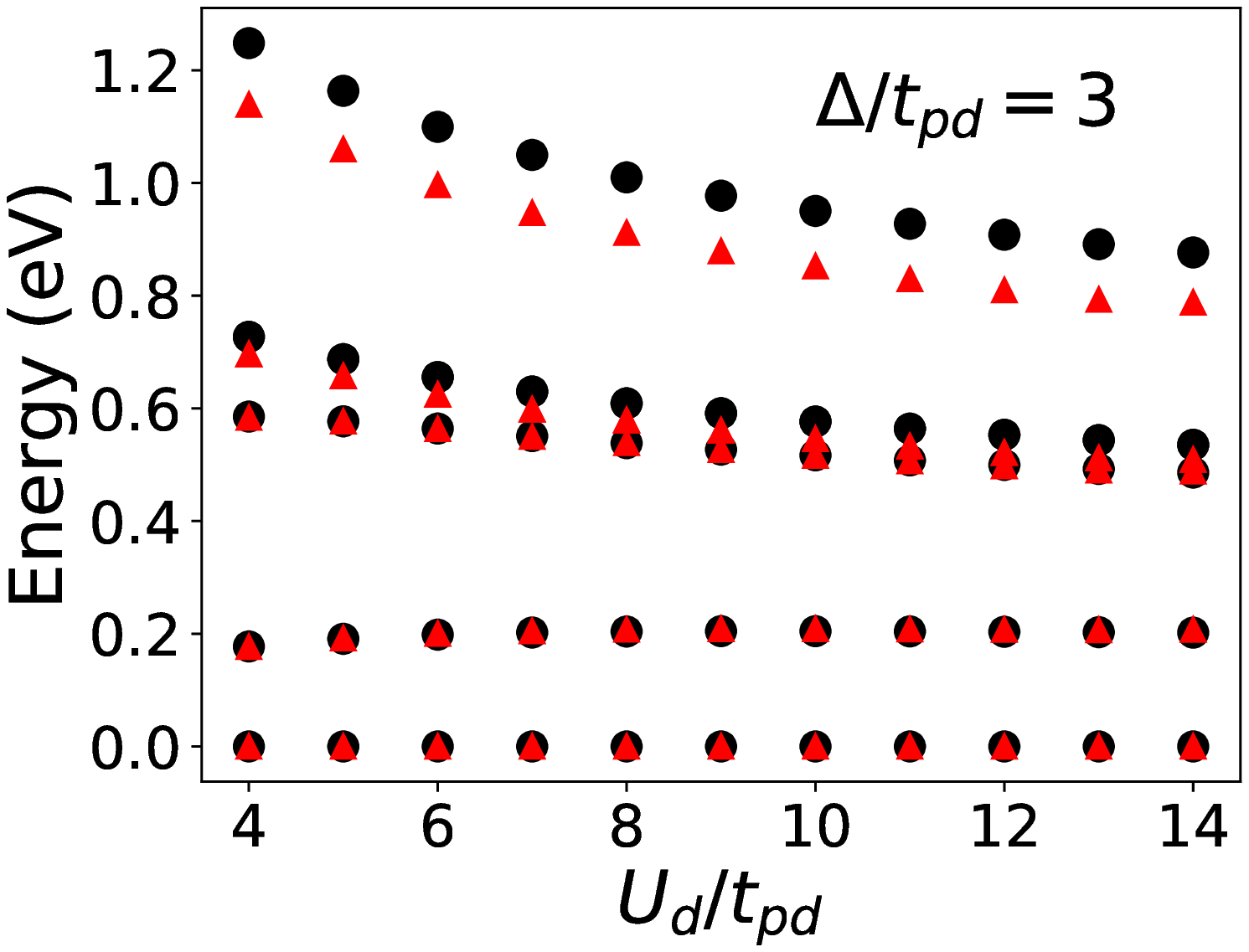}
\subfigimgtwo[width=0.49\linewidth]{(B)}{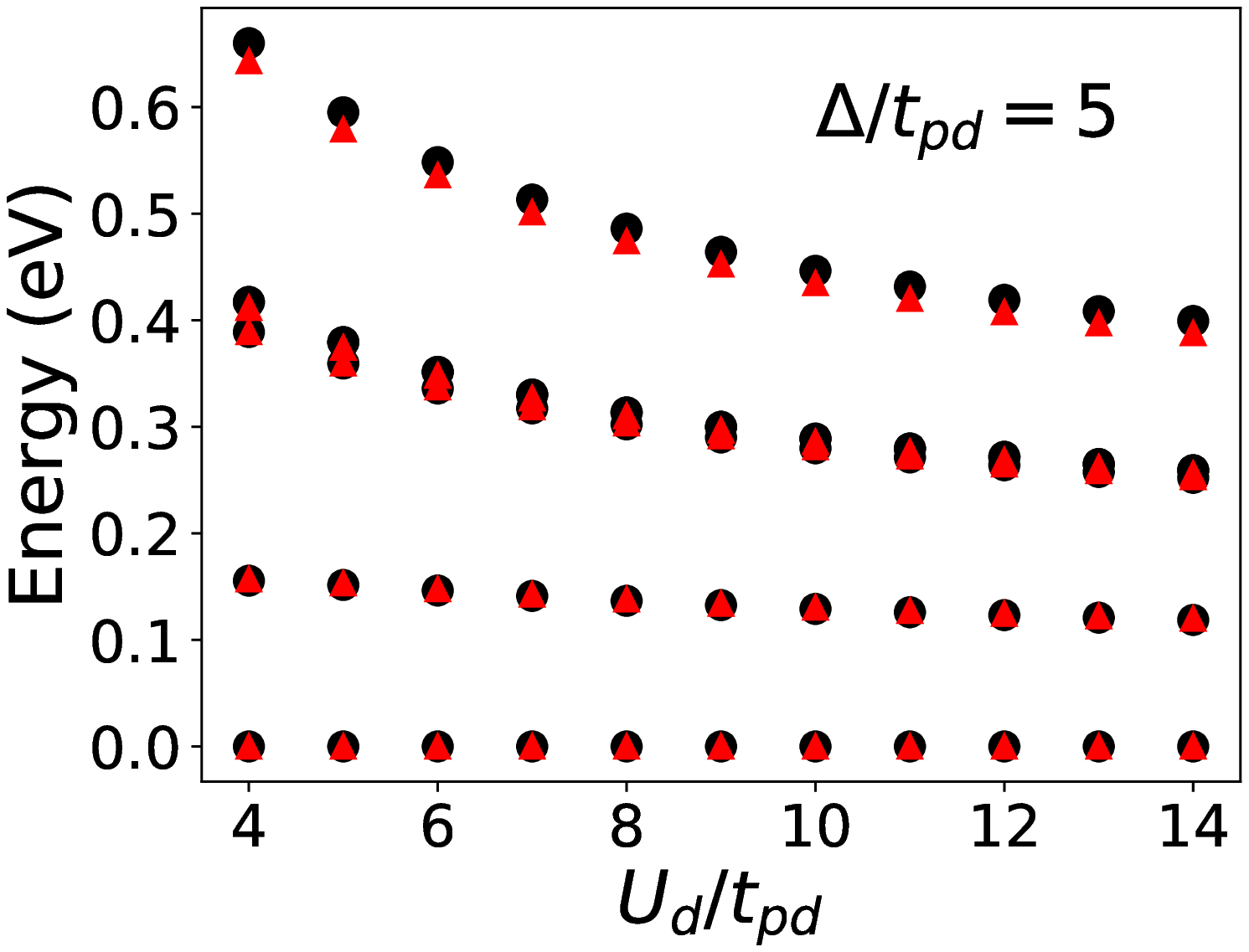}
\end{tabular}
\caption{Comparison of energy spectra of the original three-band model (black circles) and the effective one-band model (red triangles). The three-band model is on the $2\times 2$ cell. 
(A) and (B) show the energy spectra for different parameter sets of the three-band model: (A) $\Delta/t_{pd}=3$ and various $U_{d}/t_{pd}$; (B) $\Delta/t_{pd}=5$ and various $U_{d}/t_{pd}$. 
In all cases, $t_{pd}=1.3$ eV.  }
\label{fig:energyfit} 
\end{figure}	

A promise of downfolding is the reduction of the size of the effective Hilbert space; allowing 
simulations of bigger unit cells to be carried out. To show that this actually works well in practice for the three-band case, 
we consider the $2\sqrt{2} \times 2 \sqrt{2}$ square unit cell, comprising of 8 copper and 16 oxygen orbitals. 
For representative test cases, we performed exact diagonalization calculations at half filling; 
the Hilbert space comprises of 112,911,876 basis states. Roughly 200 Lanczos iterations were carried out, 
enabling convergence of the lowest four energies. We compared the lowest gaps with the 
corresponding calculation on the one-band model on the same square geometry, with a Hilbert space size of only 4,900, 
using the downfolded parameters obtained from the smaller $2 \times 2$ cell. 

Our results are summarized in Figure~\ref{fig:predictivity}. Panel (A) shows the six representative parameter sets 
of the three-band model and the corresponding downfolded one-band parameters. Panels (B) and (C) show the lowest three energy 
gaps for representative values of $U_d/t_{pd}=4,8,12$ for $\Delta/t_{pd}=3$ and $\Delta/t_{pd}=5$ respectively. 
In all cases, the agreement between the three-band and one-band models is remarkably good. The 
energy gap error of the lowest gap is within $0.0004$ eV (1\% relative error). The largest error in the third gap is 
of the order of $0.005$ eV (3\% relative error). 
These results indicate the reliability of the downfolding procedure 
and highlight its predictive power. 

\renewcommand{\subfigimgone}[3][,]{%
  \setbox1=\hbox{\includegraphics[#1]{#3}}
  \leavevmode\rlap{\usebox1}
  \rlap{\hspace*{75pt}\vspace*{100pt}\raisebox{\dimexpr\ht1-9.5\baselineskip}{#2}}
  \phantom{\usebox1}
}
\renewcommand{\subfigimgtwo}[3][,]{%
  \setbox1=\hbox{\includegraphics[#1]{#3}}
  \leavevmode\rlap{\usebox1}
  \rlap{\hspace*{75pt}\vspace*{100pt}\raisebox{\dimexpr\ht1-9.5\baselineskip}{#2}}
  \phantom{\usebox1}
}
\begin{figure}[hbt]
\begin{minipage}{0.42\linewidth}
\centering
 \leavevmode\rlap{\usebox1}
  \rlap{\hspace*{85pt}\vspace*{100pt}\raisebox{\dimexpr\ht1-5.0\baselineskip}{(A)}}
  \phantom{\usebox1}
\hbox{
\quad
\begin{tabular}{c|c||c|c}
\hline
$\Delta/t_{pd}$ & $U_d/t_{pd}$ & $t$ & $U/t$\\
\hline
\hline
3.0 & 4.0 &  0.2839 & 8.698 \\ 
3.0 & 8.0 &  0.3025 & 13.45 \\
3.0 & 12.0 & 0.3155 & 16.58 \\
5.0 & 4.0 &  0.2326 & 15.08 \\
5.0 & 8.0 &  0.2501 & 23.75\\
5.0 & 12.0 & 0.2647 & 29.89 \\
\hline
\end{tabular}}
\end{minipage}
\begin{minipage}{0.57\linewidth}
\subfigimgone[width=0.488\linewidth]{(B)}{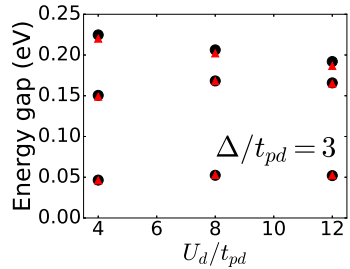}
\subfigimgtwo[width=0.488\linewidth]{(C)}{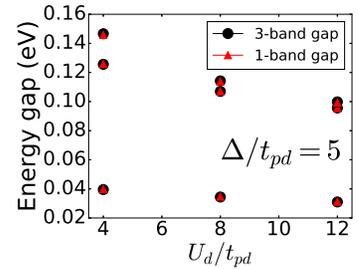}
\end{minipage}
\caption{Multiscale prediction of the effective one-band model. (A) shows the parameters of the effective one-band model obtained from DMD on a small cell, the 4-copper ($2\times2$) cell. 
These parameters were then directly used to predict the energy spectra of a larger cell, the 8-copper ($2\sqrt{2} \times 2\sqrt{2}$) cell. (B) and (C) show the predicted spectra (red triangles) 
in comparison to the exact spectra (black circles) of the three-band model of different parameters (different $\Delta/t_{pd}$ and $U_d/t_{pd}$ where $t_{pd}=1.3$ eV). }
\label{fig:predictivity}
\end{figure} 

Until this point, all our results focused on downfolding using only the lowest three eigenstates of the $2\times2$ cell. 
We now explore the effect of increasing the energy window, by including higher eigenstates, using our 
test example of $U_d/t_{pd}=8$ and $\Delta/t_{pd}=3$. To do so, we now use 
all six low energy eigenstates for optimizing the cost function in Eq.~\eqref{eq:C}. We 
find similar (but not exactly the same) values of $\alpha_i$ compared to 
the case when only the three lowest states were used. The fact that a solution with small cost can be attained 
confirms our expectation that the entire low energy space of six states is consistently 
described by a set of $\tilde{d}_i$ operators. 

\renewcommand{\subfigimgone}[3][,]{%
  \setbox1=\hbox{\includegraphics[#1]{#3}}
  \leavevmode\rlap{\usebox1}
  \rlap{\hspace*{100pt}\vspace*{12pt}\raisebox{\dimexpr\ht1-8\baselineskip}{#2}}
  \phantom{\usebox1}
}
\renewcommand{\subfigimgtwo}[3][,]{%
  \setbox1=\hbox{\includegraphics[#1]{#3}}
  \leavevmode\rlap{\usebox1}
  \rlap{\hspace*{100pt}\vspace*{12pt}\raisebox{\dimexpr\ht1-8\baselineskip}{#2}}
  \phantom{\usebox1}
}
\begin{figure}[hbt]
\centering
 \begin{tabular}{@{}p{0.90\linewidth}@{\quad}p{\linewidth}@{}}
\subfigimgone[width=0.49\linewidth]{(A)}{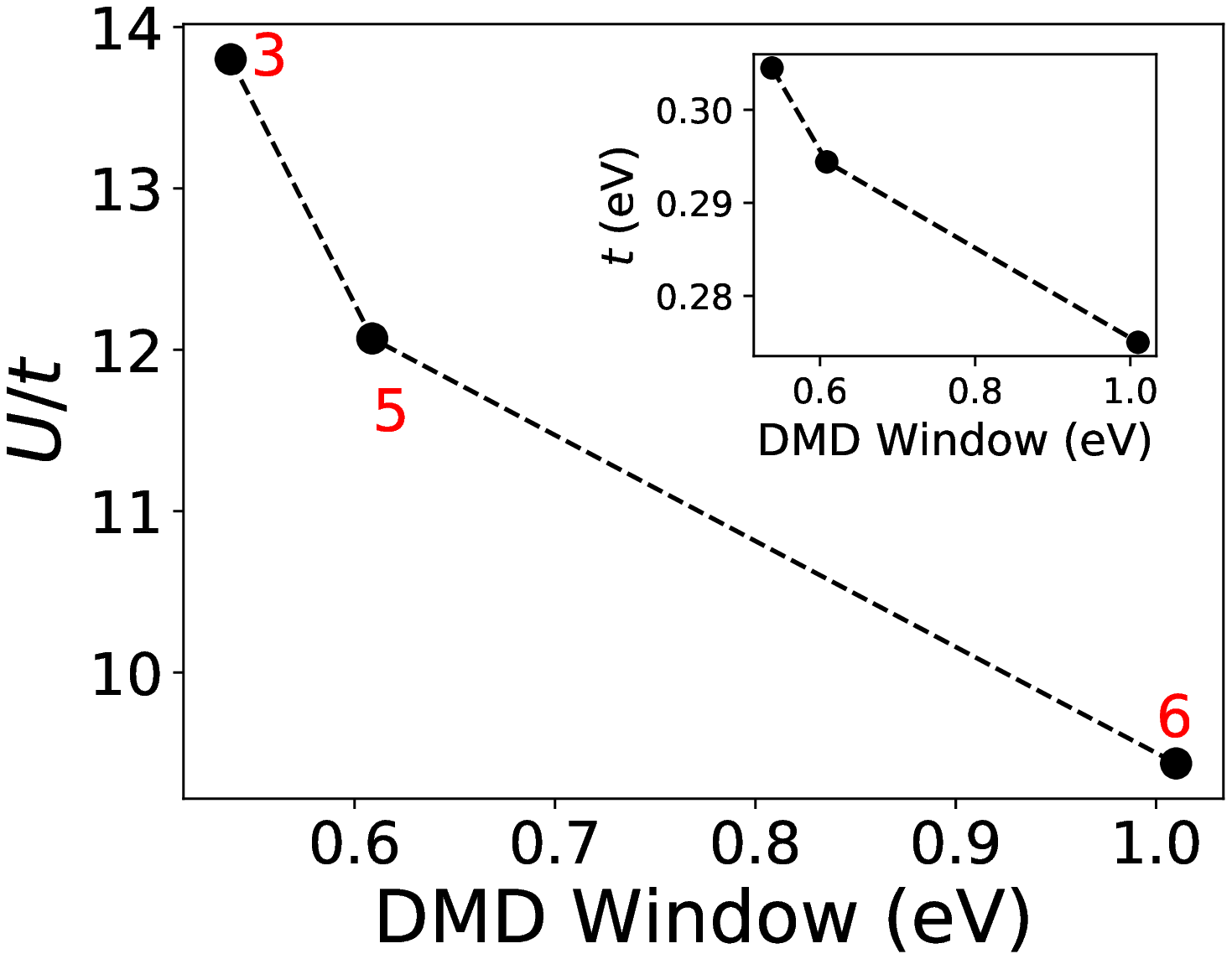}
\subfigimgtwo[width=0.50\linewidth]{(B)}{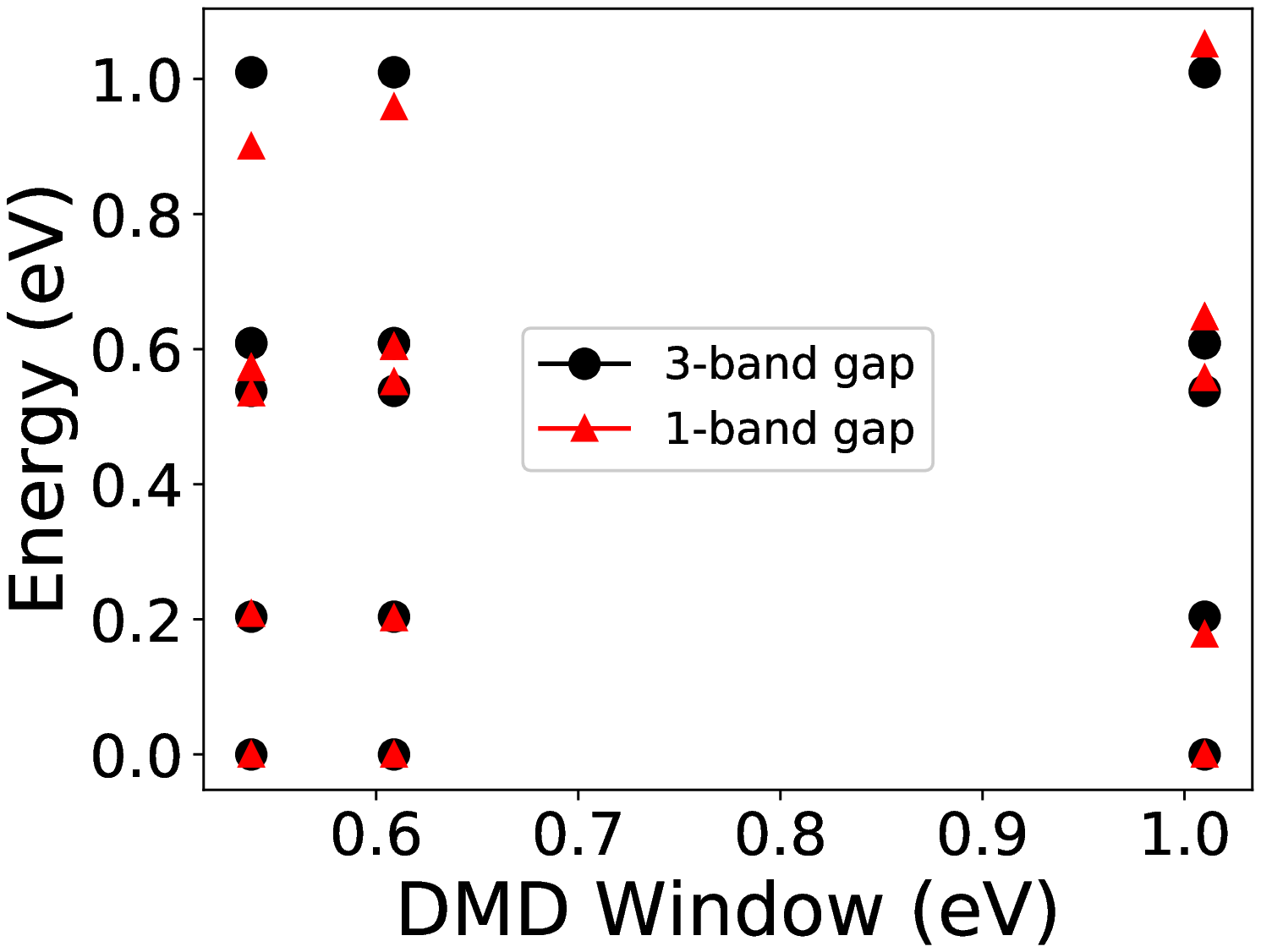}
\end{tabular}
\caption{Dependence of the downfolding parameters on the chosen low energy manifold. (A) shows the variation of downfolded parameters $U/t$ and $t$ (inset) 
with low energy manifolds of different energy windows (corresponding to 3,5 and 6 low energy eigenstates); 
(B) shows the energy spectra of the effective one-band model (red triangles) in comparison to the spectra of the original three-band model (black circles). 
The three-band Hubbard model is on the $2\times 2$ cell with $U_d/t_{pd}=8$, $\Delta/t_{pd}=3$, and $t_{pd}=1.3$.
}
\label{fig:windows} 
\end{figure}	

However, as Figure~\ref{fig:windows}(A) shows, the estimates of $U/t$ and $t$ depend on how many eigenstates 
are used in the DMD procedure. This is because the DMD aims to provide the one-band description 
that best describes \textit{all} states in a given window. If the model is not perfect within a given energy window, 
an energy dependent model is expected, consistent with the renormalization group perspective. For our test example, 
increasing the number of eigenstates from three to six changed $U/t$ from $13.8$ to $9.44$ and $t$ from $0.3045$ to $0.2750$ eV.
\footnote{When three states were used for optimizing the orbitals and for the DMD, we found $U/t\approx 13.45$ and $t=0.3025$ eV. 
This is because $\tilde{d}$ are slightly different in the two cases.} 

The features associated with the energy dependence are further confirmed in Figure~\ref{fig:windows}(B). 
which shows a comparison of energy gaps of the three-band and downfolded one-band model on the $2\times2$ cell. 
When only three states are used, the one-band (nearest neighbor) Hubbard model is insufficient for \textit{accurately}
describing states outside the window. When all six states are used, the DMD tries to minimize the error of the 
largest energy gap at the cost of errors in the smaller energy gaps. 
One could of course choose a different parameterization, say with additional next nearest neighbor $t'$, for which is 
may be possible to reduce this energy dependence significantly and thus have a model that describes the smaller 
and larger energy scales equally well.

\subsection{One dimensional hydrogen chain}
\label{subsection:1dhydrogen}
We now move on to one of the simplest extended \emph{ab initio} systems, a hydrogen chain in one dimension with periodic boundary conditions. The one-dimensional hydrogen chain has been used as a model for validating a variety of modern \textit{ab initio} many-body methods \cite{H10_Simons}. 
We consider the case of $10$ atoms with periodic boundary conditions and work in a regime where the inter-atomic distance $r$ is in the range $1.5 - 3.0$ \AA, such that the system is  well described in terms of primarily $s$-like orbitals. 

For a given $r$, we first obtain single-particle Kohn-Sham orbitals from a set of spin-unrestricted and 
spin-restricted DFT-PBE calculations. The localized orbital basis upon which the RDMs (descriptors) 
are evaluated is obtained by generating intrinsic atomic orbitals (IAO) \cite{knizia_intrinsic_2013} from the Kohn-Sham orbitals 
orthogonalized using the L\"owdin procedure (see Figure~\ref{fig:fit_quality}). These are the orbitals that enter the one-band Hubbard Hamiltonian. 
Then, to generate a database of wavefunctions needed for the DMD, we produce a set of Slater-Jastrow 
wavefunctions consisting of singles and doubles excitations to the Slater determinant:
\begin{subequations}
\begin{eqnarray}
| s \rangle = & e^J \Big[a^\dagger_{i \eta} a_{k \eta}   | KS \rangle \Big] \,,\\
| d \rangle = & \: e^J \Big[a^\dagger_{i \eta} a^\dagger_{j \eta'} a_{k \eta'} a_{l \eta}   | KS \rangle\Big] ,
\end{eqnarray}
\end{subequations}
where $|KS\rangle$ is the Slater determinant of occupied Kohn-Sham orbitals, $\eta \neq \eta'$ are spin indices, 
and $a_{i}^\dagger$ ($a_{i}$) is a single-electron creation (destruction) operator corresponding to a particular Kohn-Sham orbital. The $k,l$ indices label occupied orbitals in the original Slater determinant, while $i,j$ are virtual orbitals. 
$e^J$ is a Jastrow factor optimized by minimizing the variance of the local energy. 

We compute the energies (expectation values of the Hamiltonian) and the RDMs for each wave function within DMC. 
By computing the trace of the resulting 1-RDMs, we verify that all the electrons present in the system are represented within the localized basis of $s$-like orbitals. 
If the trace of the 1-RDM deviates from the nominal number of electrons for a particular state by more than some chosen threshold - 2\% in this example -
it indicates that some orbitals are occupied ($2s$- or $2p$-like orbitals for hydrogen)
that are not represented within the localized IAO basis used for computing the descriptors. 
Hence, these states do not exist within the $\mathcal{LE}$ space, and cannot be described by a one-band $s$-orbital model. We exclude such states from the wave function set. 
The acquired data is then used in DMD to downfold to a one-band Hubbard Hamiltonian.
\renewcommand{\subfigimg}[3][,]{%
  \setbox1=\hbox{\includegraphics[#1]{#3}}
  \leavevmode\rlap{\usebox1}
  \rlap{\hspace*{45pt}\vspace*{12pt}\raisebox{\dimexpr\ht1-5.5\baselineskip}{#2}}
  \phantom{\usebox1}
}
\begin{figure}[hbt]
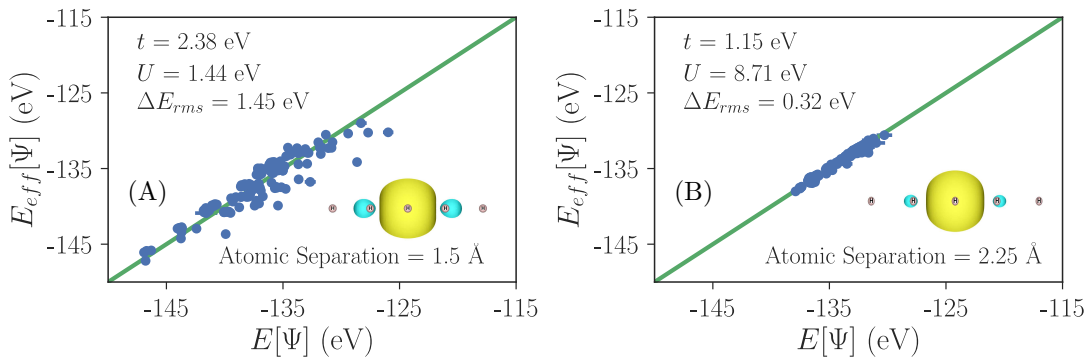

\begin{tabular}{@{}p{0.95\linewidth}@{\quad\quad}p{}@{}}
\centering
   \subfigimg[width=0.45\linewidth]{(A)}{{H_chain_fit_model_length1.5_tUs_inset}.eps}
   \subfigimg[width=0.45\linewidth]{(B)}{{H_chain_fit_model_length2.25_tUs_inset}.eps}\\
 \end{tabular}
\caption{Reconstructed model energy ($E_{eff}[\psi]$) versus DMC energy ($E[\psi]$) for the H$_{10}$ chain at (A) 1.5 \AA \: and (B) 2.25 \AA \:. 
The energy range of excitations narrows significantly for larger interatomic separation. Insets show the intrinsic atomic orbitals which constitute the one-body space 
which was used for calculating the reduced density matrices (descriptors).  
\label{fig:fit_quality}
  }
\end{figure}

Figure~\ref{fig:fit_quality} shows the fitting results of the energy functional $E[\Psi]$ within the sampled $\mathcal{LE}$ for two representative distances (1.5 and 2.25\AA). As we can see, the model $E_{eff}[\Psi]$ reproduces the \textit{ab initio} $E[\Psi]$ up to certain error that decreases with atomic separation. That is, the fitted Hubbard model provides a more accurate description as separation distance increases, and the system becomes more atomic-like. 

Figure~\ref{fig:Parameters-vs-Bond-t} shows the fitted values of the downfolding parameters $t$ and $U/t$ at various distances. 
$t$ decreases as the interatomic distance increases, and the value of $U/t$ increases. The single-band Hubbard model qualitatively captures how the system approaches the atomic limit, in which $t$ becomes zero. 
 
\renewcommand{\subfigimg}[3][,]{%
  \setbox1=\hbox{\includegraphics[#1]{#3}}
  \leavevmode\rlap{\usebox1}
  \rlap{\hspace*{30pt}\vspace*{20pt}\raisebox{\dimexpr\ht1-5.0\baselineskip}{#2}}
  \phantom{\usebox1}
}
\begin{figure}[hbt]
   \centering
 \begin{tabular}{@{}p{1.00\linewidth}@{}p{\linewidth}@{}}
   \centering
    \subfigimg[width=0.31\linewidth]{(A)}{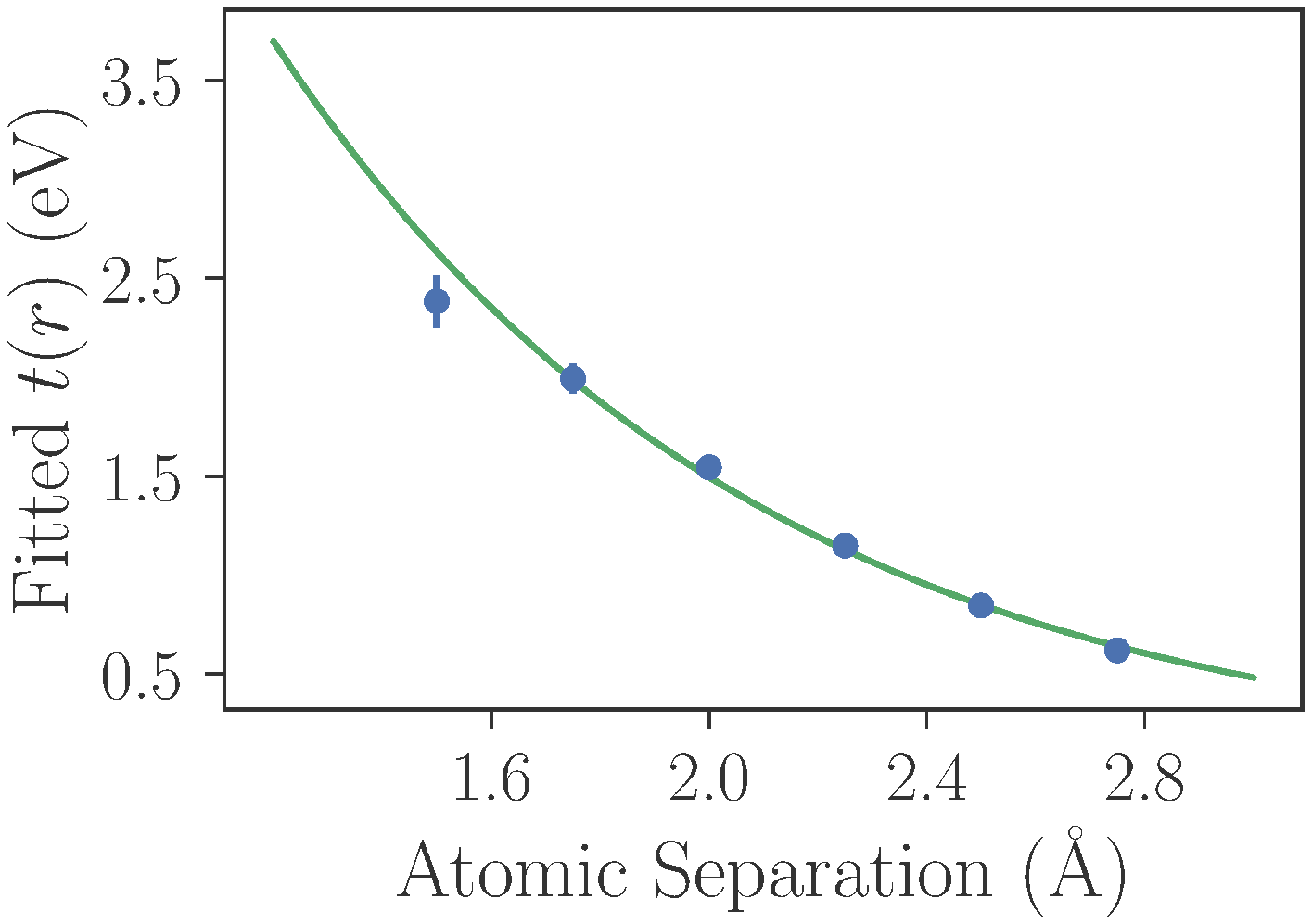}
    \subfigimg[width=0.31\linewidth]{(B)}{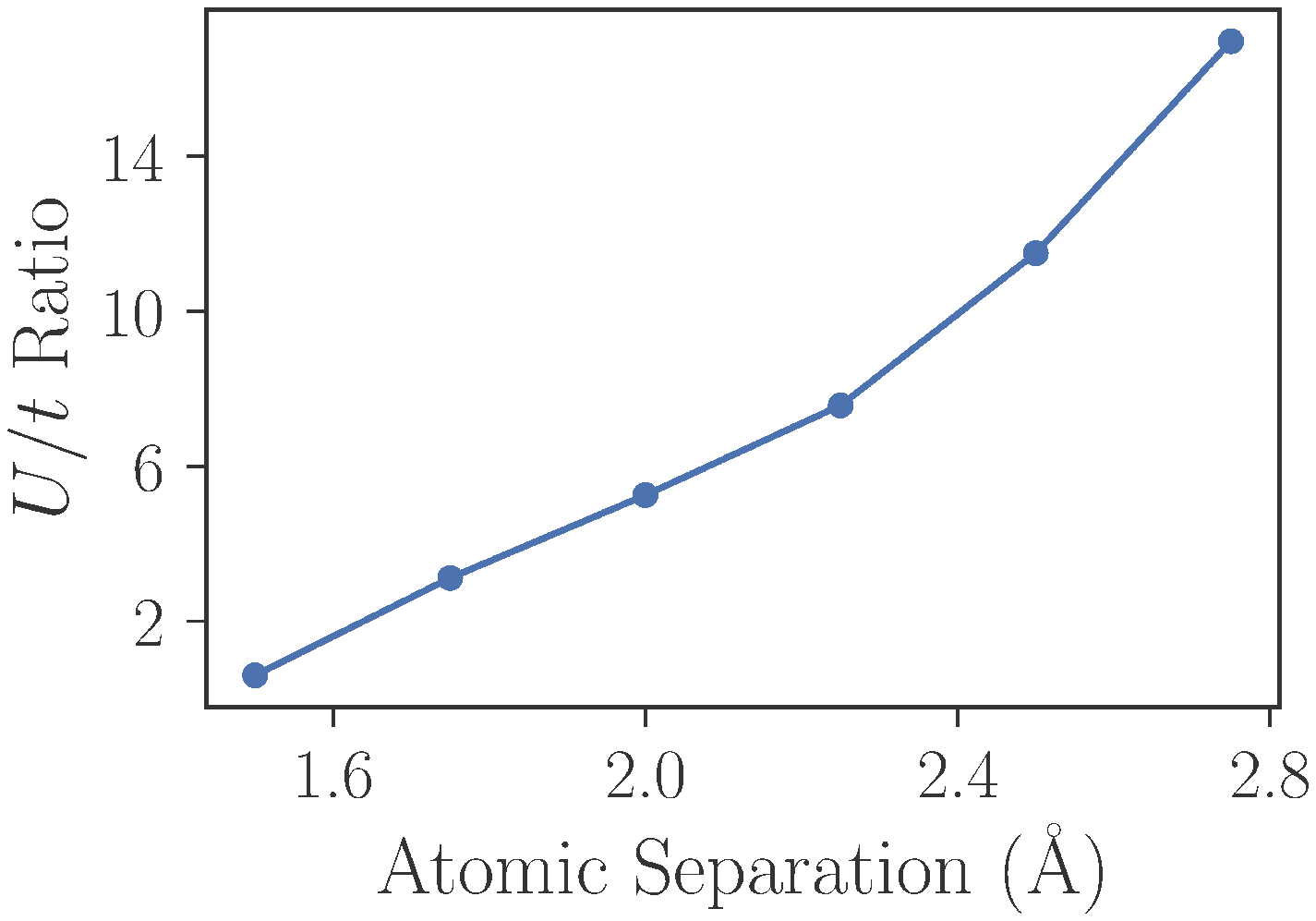}
    \subfigimg[width=0.32\linewidth]{(C)}{{r2_ut_vs_separation_h_chain}.eps}
 \end{tabular}
\caption{ (A) The one-body hopping $t$ parameter as a function of interatomic distance for the periodic H$_{10}$ chain, obtained from a fitted $U$-$t$ model. $t$ declines to zero as $r$ increases. 
(B) The ratio $U/t$ for the fitted parameter values as a function of interatomic separation. The ratio is small at lower bond-lengths, where $t$ is more relevant in describing the system, and larger at longer bond-lengths, where inter-site hopping is less significant. 
(C) The R$^2$ fit parameters obtained from fitting the $U$-$t$ model to the H$_{10}$ chain, as a function of interatomic separation. }\label{fig:Parameters-vs-Bond-t}
\end{figure}

The R$^2$ values obtained from fitting the descriptors to the \textit{ab initio} energy [see Figure~\ref{fig:Parameters-vs-Bond-t}(C)] also show that the single-band Hubbard model is a good description of the system at large distances, but not at small distances. 
This is primarily because the dynamics of other degrees of freedom (e.g. $2s$ and $2p$ orbitals) become important to the low energy spectrum at small distances. Other interaction terms beyond the on-site Hubbard $U$, such as nearest-neighbor Coulomb interactions and Heisenberg coupling, can also become significant. 
Without including higher orbitals or additional many-body interaction terms, the model gives rise to an incorrect insulator state at small distances. 
Conversely, at larger separations ($r>1.8$\AA), where the system is in an insulator phase \cite{Stella2011}, the model provides a better description. 
\subsection{Graphene and hydrogen honeycomb lattice}
\label{subsection:graphene}
Our third example highlights the role of the high energy 
degrees of freedom not present in the low energy description 
but which are instrumental in renormalizing the effective interactions. 
We demonstrate this by considering the case of graphene, and by 
comparing it to artificially constructed counterparts without the high energy electrons. 
Although many electronic properties of graphene can be adequately 
described by a noninteracting tight-binding model of $\pi$ electrons~\cite{Castro2009}, 
electron-electron interactions are crucial for explaining 
a wide range of phenomena observed in experiments~\cite{Kotov2012}. 
In particular, electron screening from $\sigma$ bonding renormalizes 
the low energy plasmon frequency of the $\pi$ electrons~\cite{Zheng2016}. 
In fact a system of $\pi$ electrons with bare Coulomb interactions has been shown to be an insulator instead of a semimetal~\cite{DrutPRL2009, DrutPRB2009,  Smith2014, Zheng2016}. 
Using DMD, we demonstrate how the screening effect of $\sigma$ electrons is manifested in the low energy effective model of graphene. 

In order to disentangle the screening effect of $\sigma$ electrons from the bare interactions 
between $\pi$ electrons, we apply DMD to three different systems, graphene, $\pi$-only graphene, and a honeycomb lattice of hydrogen atoms.  
In the $\pi$-only graphene, the 
$\sigma$ electrons are replaced with a static constant negative charge background. 
The role of $\sigma$ electrons is then clarified by comparing the effective model Hamiltonians of these two systems. 
The hydrogen system we study has the same lattice constant $a=2.46$~\AA~as graphene, 
which has a similar Dirac cone dispersion as graphene~\cite{Zheng2016}. 

By constructing the one-body space by Wannier localizing Kohn-Sham orbitals obtained from DFT calculations (see Figure~\ref{fig:honeycomb_wan}), 
we verify that the low energy degrees of freedom correspond to the $\pi$ orbitals in graphene and 
its $\pi$-only system and $s$ orbitals in hydrogen; these enter the effective one-band Hubbard model description in Eq.~\eqref{eq:oneband}. 
Due to the vanishing density of states at the Fermi level, the Coulomb interaction remains long-ranged, 
in contrast to usual metals where the formation of electron-hole pairs screens the interactions strongly~\cite{Zheng2016}. 
However, for certain aspects, the long ranged part can be considered as renormalizing the 
onsite Coulomb interaction $U$ at low energy~\cite{Schuler2013, Changlani2015}. 

\renewcommand{\subfigimg}[3][,]{%
  \setbox1=\hbox{\includegraphics[#1]{#3}}
  \leavevmode\rlap{\usebox1}
  \rlap{\hspace*{20pt}\vspace*{18pt}\raisebox{\dimexpr\ht1-1.27\baselineskip}{#2}}
  \phantom{\usebox1}
}
\begin{figure}[hbt]
\centering
 \begin{tabular}{@{}p{0.90\linewidth}@{\quad}p{\linewidth}@{}}
   \subfigimg[clip, width=0.45\textwidth]{(A)}{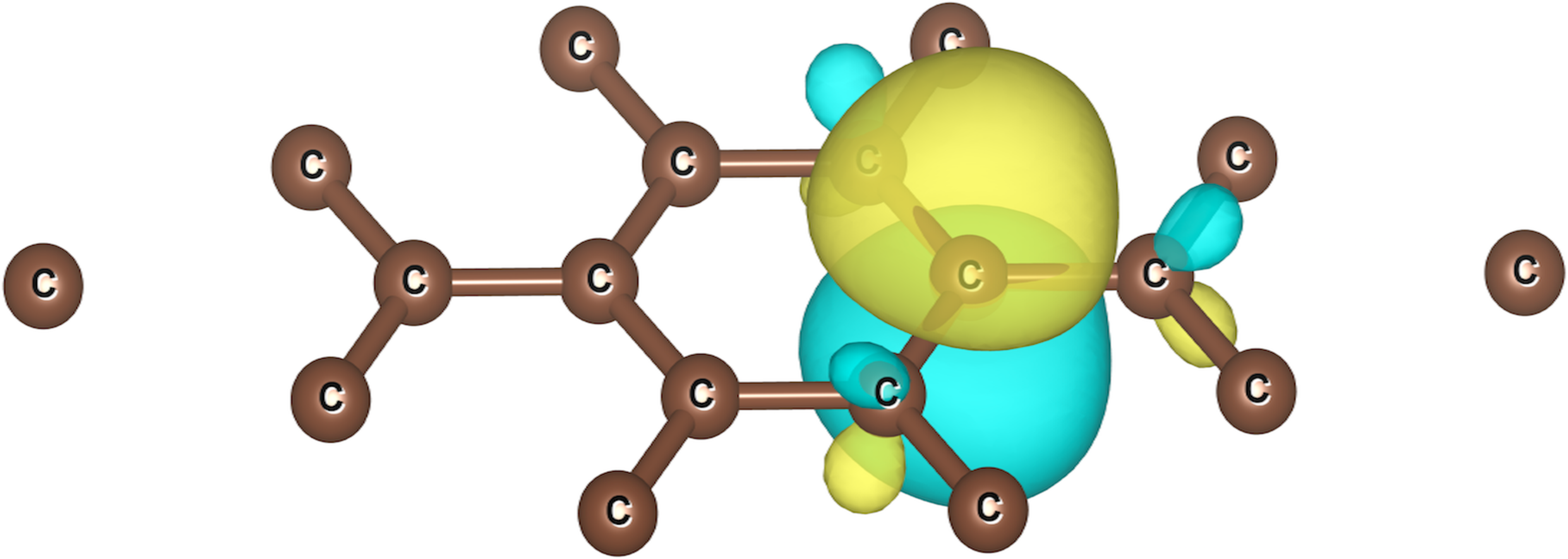}
   \subfigimg[clip, width=0.45\textwidth]{(B)}{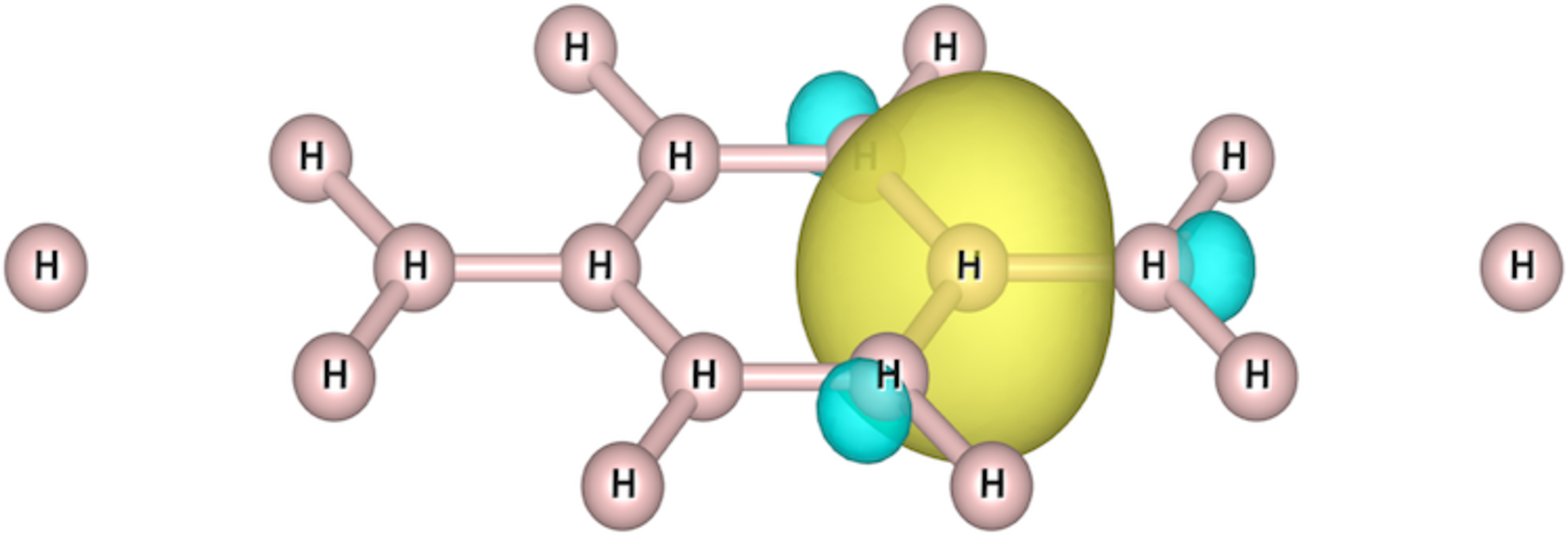}
 \end{tabular}
\caption{Wannier orbitals constructed from Kohn-Sham orbitals: (A) graphene $\pi$ orbital; (B) hydrogen $s$ orbital. }
\label{fig:honeycomb_wan}
\end{figure}

To estimate the one-band Hubbard parameters, we used the DMD method using a set of 50 Slater-Jastrow wave functions that correspond 
to the electron-hole excitations within the $\pi$ channel for the graphene systems 
or $s$ channel for the hydrogen system. In particular, for graphene, 
the Slater-Jastrow wave functions are constructed from occupied $\sigma$ bands and occupied $\pi$ bands, whereas for $\pi$-only 
graphene, Slater-Jastrow wave functions constructed from occupied $\pi$ Kohn-Sham orbitals of graphene. The \textit{ab initio} simulations 
were performed on a $3\times3$ cell (32 carbons or hydrogens) and the energy and RDMs of these wave functions were
evaluated with VMC. The error bars on our downfolded parameters are estimated using the jackknife method \cite{Jackknife1981}.
The results from our calculations are summarized in Figure~\ref{fig:ne_aidmd_gh}.

\renewcommand{\subfigimg}[3][,]{%
  \setbox1=\hbox{\includegraphics[#1]{#3}}
  \leavevmode\rlap{\usebox1}
  \rlap{\hspace*{42pt}\vspace*{12pt}\raisebox{\dimexpr\ht1-1.2\baselineskip}{#2}}
  \phantom{\usebox1}
}
\begin{figure}[hbt]
\centering
  \begin{tabular}{@{}p{0.95\linewidth}@{\quad}p{\linewidth}@{}}
    \subfigimg[clip, width=0.325\linewidth]{(A)}{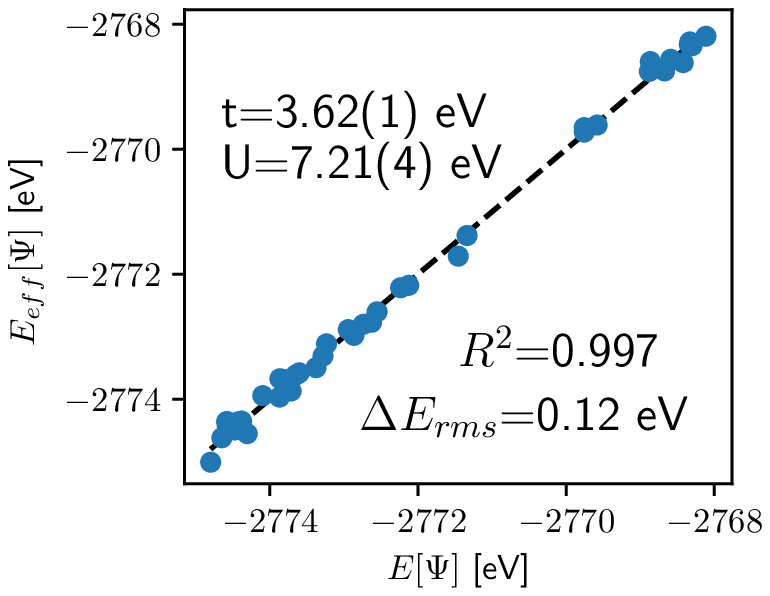}
    \subfigimg[clip, width=0.325\linewidth]{(B)}{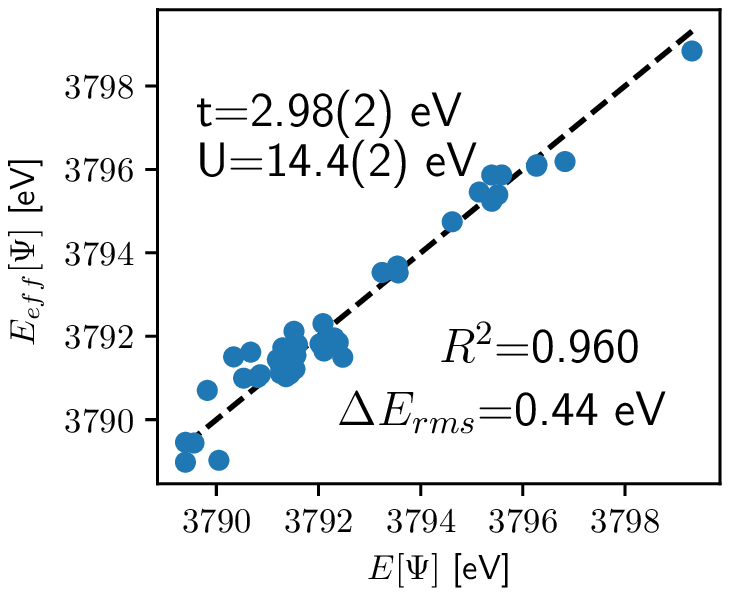}
    \subfigimg[clip, width=0.325\linewidth]{(C)}{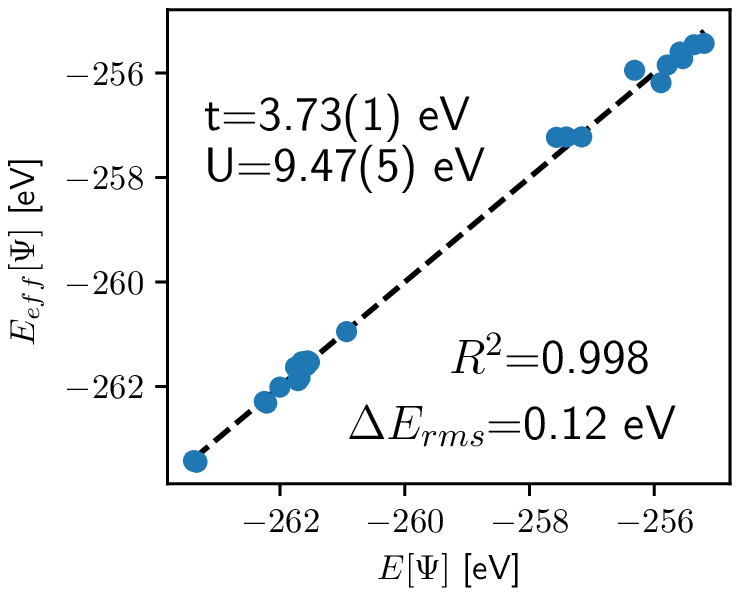}
    \end{tabular}
\caption{Comparison of \textit{ab initio} ($E[\psi]$) and fitted energies ($E_{eff}[\psi]$) 
of the 3$\times$3 periodic unit cell of graphene and hydrogen lattice: (A) graphene; (B) $\pi$-only graphene; (C) hydrogen lattice.}\label{fig:ne_aidmd_gh}
\end{figure}

We find that the one-band Hubbard model describes graphene and hydrogen very well, as is seen from the fact that $R^2$ is closed to 1 for the fits. Our fits are shown in Figure~\ref{fig:ne_aidmd_gh}.
For both graphene and hydrogen, $U/t$ is smaller than the critical value of the 
semimetal-insulator transition $(U/t)_c \approx 3.8$ for the honeycomb lattice~\cite{Sorella2012}, 
which is consistent with both systems being semimetals. The two systems indeed have similar hopping constant $t$, 
consistent with the fact that they have similar Fermi velocities at the Dirac point. However, 
the difference in their high energy structure manifests itself as differently renormalized electron-electrons interactions, 
explaining the difference in $U$. Most prominently, the $\pi$-only system has much larger $U/t$ ($\sim4.9$) compared to graphene, 
which is large enough to push it into the insulating (antiferromagnetic) phase.
Thus, downfolding shows the clear significance of $\sigma$ electrons in renormalizing the effective onsite interactions of the $\pi$ orbitals,making graphene a weakly interacting semimetal instead of an insulator.  

\subsection{FeSe diatomic molecule}
\label{subsection:fese}
Transition metal systems are often difficult to model due to the many orbital and possibly magnetic descriptors introduced by $d$ electrons. 
This is seen in the proliferation of models for transition metals, which include terms like spin-spin coupling, spin-orbital coupling, hopping, Hund's like coupling, and so on. 
Models containing all possible descriptors are unwieldy, and it is difficult to determine which degrees of freedom are needed for a minimal model to reproduce an interesting effect. 
Transition metal systems are challenging to describe using most electronic structure methods because of the strong electron correlations and multiple oxidation states possible in these systems. 
Fixed-node DMC has been shown to be a highly accurate method on transition metal materials in improving the description of the ground state properties and energy gaps~\cite{Foyevtsova2014, Wagner_Abbamonte, Zheng2015, Wagner2016}. In this section, we apply DMD using fixed-node DMC to quantify the importance of various interactions in a FeSe diatomic molecule with a bond length equal to that of the iron based superconductor, FeSe~\cite{kumar_crystal_2010}, in order to help identifying the descriptors that may be relevant in the bulk material.
\begin{figure*}[htb]
  \centering
  \includegraphics[width=0.8\textwidth]{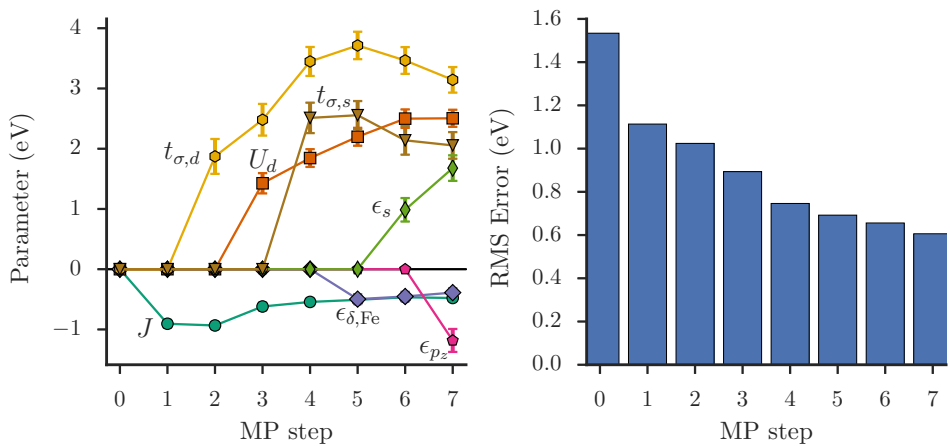}
  \caption{
    \label{fig:fese} 
    (Left) Parameter values for each fit generated in the MP algorithm, labeled at the step where they are included in the model. 
    A zero value indicates that parameter is not yet added to the model.
    The sign of $J$ is consistent with Hund's rules, and the signs of $t_{\sigma,d}$ and $t_{\sigma,s}$ are consistent with Se being located in positive $z$ with respect to Fe. 
    (Right) RMS error of each model generated by MP as the algorithm includes parameters. 
    The RMS error of the largest model considered was 0.61~eV.
  }
\end{figure*}

We considered a low-energy space spanned by the Se $4p$, Fe $3d$, and Fe $4s$ orbitals. We sampled singles and doubles excitations from a reference Slater determinant 
of Kohn-Sham orbitals taken from DFT calculations with PBE0 functional with total spin 0, 2, and 4, which were then multiplied by a Jastrow factor and further optimized using fixed-node DMC. 
After this procedure, 241 states were within a low energy window of 8 eV. 
Of these, eight states had a significant iron $4p$ component, which excludes them from the low-energy subspace. 
This leaves us with 233 states in the low-energy subspace.

We consider a set of 21 possible descriptors consisting of local operators on the iron $4s$, iron $3d$ states, and selenium $4p$ states, which is a total of 9 single-particle orbitals.
We use the same IAO construction as Section~\ref{subsection:1dhydrogen} to generate the basis for these operators.
At the one-body level, we consider orbital energy descriptors: 

\begin{align}
  &\epsilon_{s} n_s,&
  &\epsilon_{\pi,\mathrm{Se}} (n_{p_x} + n_{p_y}), &
  &\epsilon_{z} n_{p_z},&
  \nonumber \\
  &\epsilon_{z^2} n_{d_{z^2}},& 
  &\epsilon_{\pi,\mathrm{Fe}} (n_{d_{xz}} + n_{d_{yz}}).& 
  &\epsilon_{\delta} (n_{d_{xy}} + n_{d_{x^2-y^2}}),&
\end{align}

and the symmetry-allowed hopping terms:

\begin{align}
  &t_{\sigma,d} \sum_{\eta} \left( c_{d_{z^2},\eta}^{\dagger} c_{p_z,\eta} + \text{h.c.} \right),&
  &t_{\sigma,s} \sum_{\eta} \left(c_{s,\eta}^{\dagger}  c_{p_z,\eta} + \text{h.c.} \right),&
  &t_{\pi} \sum_{\eta} \left(c_{d_{xz},\eta}^{\dagger}  c_{p_x,\eta} + c_{d_{yz},\eta}^{\dagger}  c_{p_y,\eta} + \text{h.c.} \right).&
\end{align}

As before, $\eta$ represents the spin index. At the two-body level, we consider Hubbard interactions:

\begin{align}
  &U_p \sum_{i \in p} n_{i,\uparrow} n_{i,\downarrow},&
  &U_{d,\delta} \sum_{i\in \{d_{xy},d_{x^2-y^2}\}} n_{i,\uparrow} n_{i,\downarrow},&
  \nonumber \\
  &U_d \sum_{i \in d} n_{i,\uparrow} n_{i,\downarrow},&
  &U_{d,\pi} \sum_{i\in \{d_{xz},d_{yz}\}} n_{i,\uparrow} n_{i,\downarrow},&
  &U_{d_{z^2}} n_{d_{z^2},\uparrow} n_{d_{z^2},\downarrow},&
\end{align}

where $p$ refers to the Se-$4p$ orbitals and $d$ refers to the Fe-$3d$ orbitals. 
Importantly, we also account for the Hund's coupling terms for the iron atom:

\begin{align}
  &J \sum_{\substack{i\ne j \\i,j \in d}} S_i \cdot S_j,&
  &J_{\delta} S_{d_{xy}} \cdot S_{d_{x^2-y^2}},&
  &J_{\delta,d_{z^2}} (S_{d_{xy}} + S_{d_{x^2-y^2}}) \cdot S_{d_{z^2}},& \label{eqn:hund1}
  \nonumber \\
  &J_{\pi} S_{d_{xz}} \cdot S_{d_{yz}},&
  &J_{\pi,d_{z^2}} (S_{d_{xz}} + S_{d_{yz}}) \cdot S_{d_{z^2}}.&
  &J_{\pi,\delta} (S_{d_{xz}} + S_{d_{yz}}) \cdot (S_{d_{xy}} + S_{d_{x^2-y^2}}),&
\end{align}

Finally, we also add a nearest neighbor Hubbard interaction: $V \sum_{i\in p, j\in d} n_{i} n_j$.

To generate a minimal description of the system, we employed a matching pursuit (MP) method~\cite{MP_Zhang1993}.
MP chooses to add descriptors based on their correlation with the residual of the linear fit. 
We started with a model that only consists of $E_0$. The Hund's coupling descriptor [first term in Eq.~\eqref{eqn:hund1}]
has the largest correlation coefficient with the residual fit, so it is added first. The fact that the Hund's coupling is chosen first in MP 
is consistent with the several studies in the literature, which find a prominent Hund's coupling can explain some 
of the properties of bulk FeSe.~\cite{demedici_hunds_2011,de_medici_janus-faced_2011,georges_strong_2013,busemeyer_competing_2016}. 
Next, MP includes the descriptor that correlates most strongly with the residuals of this first minimal model, in this case the hopping between $d$ and $p$ $\sigma$-symmetry orbitals. 
We repeated this procedure until the RMS error did not improve more than 0.05 eV upon adding a new parameter.
This criterion was chosen to strike a balance between the complexity of the model and the accuracy in reproducing the sample set.

The following model was produced:
\begin{eqnarray}
  H_{eff} &=& \epsilon_{\delta,\mathrm{Fe}} (n_{d_{xy}} + n_{d_{x^2-y^2}}) + \epsilon_s n_{s}+\epsilon_{z} n_{p_z} \nonumber \\
          &&+ t_{\sigma,d} \sum_{\eta} \left( c_{d_{z^2},\eta}^{\dagger} c_{p_z,\eta} + \text{h.c.} \right)+t_{\sigma,s} \sum_{\eta} \left(c_{s,\eta}^{\dagger}  c_{p_z,\eta} + \text{h.c.} \right)
              \nonumber \\
          &&+ U_d \sum_{i \in d} n_{i,\uparrow} n_{i,\downarrow} + J \sum_{\substack{i\ne j \\i,j \in d}} S_i \cdot S_j + E_0. \label{eq:fesemodel}
\end{eqnarray}
As before, $\eta$ is the spin index and $i$ is the orbital index, and $d$ is the set of iron $3d$ orbitals, as above.
$E_0$ is an overall energy shift, also included as a fit parameter.
The parameter values and corresponding error of each model produced by MP are shown in Figure~\ref{fig:fese}.
Note that all parameters may change at each step because the entire model is refitted when an addition parameter is included in each iteration.
The parameters are smoothly varying with the inclusion of new parameters, and they take the correct signs based on symmetry (where applicable). 
The RMS error decreases with each additional parameter, but less so as the algorithm appends additional parameters. 
Eventually the diminishing improvements do not merit the additional complexity of more parameters.

\section{Conclusion and Future prospects}
\label{sec:conclusion}
The density matrix downfolding (DMD) technique uses data derived from low-energy approximate solutions to a high energy Hamiltonian to systematically determine an effective Hamiltonian that describes the low-energy behavior of the system.
It is based on several rather simple proofs which occupy a role similar to the variational principle; they allow us to know which effective models are closer to the correct solution than others. 
The method is very general and does not require a quasiparticle picture to apply, and neither does it have double-counting issues.
It treats all interactions on an equal footing, so hopping parameters are naturally modified by interaction parameters and so on.
While most of the applications have used the first principles quantum Monte Carlo method to obtain the low-energy solutions, the method is completely general and can be used with any solution method that can produce high quality energy and reduced density matrices.  
We have discussed several examples to present the conceptual and algorithmic aspects of DMD. 

The resultant lattice model can be efficiently and accurately solved for large system sizes~\cite{LeBlanc_PRX} using techniques designed and suited for small local Hilbert spaces; these include exact or selected diagonalization~\cite{DeRaedt,Tubman_selci,Holmes_Tubman_Umrigar}, density matrix renormalization group (DMRG)~\cite{White1992}, tensor networks~\cite{PEPS,Changlani_CPS,NeuscammanCPS}, dynamical mean field theory (DMFT)~\cite{Kotliar2006}, density matrix embedding (DMET)~\cite{DMET_2012} and lattice QMC methods~\cite{Scalapino, Trivedi_Ceperley, Zhang_AFQMC, Sandvik_loops, Prokofiev, 
Booth2009,SQMC,Holmes_Changlani_Umrigar, Booth2013}. 
These methods have also been used to obtain excited states, dynamical correlation functions and thermal properties, that have been difficult to obtain in \textit{ab initio} approaches.

DMD, though conceptually simple, is still a method in its development stages, with room for algorithmic improvements and new applications. 
Advances in the field of inverse problems~\cite{Berg2017} could be incorporated into DMD to 
mitigate the problems associated with optimization and over-fitting. 
Here we briefly outline some aspects that need further research:
\begin{enumerate}
	\item The wave function database ($\ket{\Psi} \in \mathcal{LE}$):
	The DMD method relies crucially on the availability of a low energy space of \textit{ab initio} wave functions. While these wave functions do not have to be eigenstates, automating their construction remains challenging and realistically requires knowledge of the physics to be described.
	\item Optimal choice of basis functions. The second-quantized operators in the effective Hamiltonian correspond to a basis in the continuum. The quality of the model depends on the basis describing the changes between low-energy wave functions accurately.
	\item Form of the low energy model Hamiltonian. While the exact effective Hamiltonian is unique, there may be many ways of approximating it with varying levels of compactness and accuracy.
\end{enumerate} 
The advantage of the DMD framework is that all these can be resolved internally.
Given a good sampling of $\mathcal{LE}$, (2) and (3) can be resolved using regression. 
Given that (2) and (3) are correct or near correct, then (1) can be resolved by finding binding planes, as noted in Section~\ref{sec:theory}.
The method thus has a degree of self consistency; it will return low errors only when 1-3 are correct.

We have shown applications to strongly correlated models (3-band), {\it ab initio} bulk systems hydrogen chain and graphene, and a transition metal molecule FeSe.
The technique is on the verge of being applied to transition metal bulk systems; there are no major barriers to this other than a polynomially scaling computational cost and the substantial amount of work involved in parameterizing and fitting models to these systems.
Looking into the future, we anticipate that this technique can help with the definition of a correlated materials genome--what effective Hamiltonian best describes a given material is highly relevant to its behavior.

\section*{Acknowledgments} 
We thank  David Ceperley,  Richard Martin, Cyrus Umrigar,  Garnet Chan,  Shiwei Zhang, Steven White,  
Lubos Mitas, So Hirata, Bryan Clark, Norm Tubman, Miles Stoudenmire and Victor Chua for extremely useful and insightful discussions. 
This work was funded by the grant DOE FG02-12ER46875 (SciDAC). 
HZ acknowledges support from Argonne Leadership Computing Facility, a U.S. Department of Energy, Office of Science User Facility under Contract DE-AC02-06CH11357.
HJC acknowledges support from the U.S. Department of Energy, 
Office of Basic Energy Sciences, Division of Materials Sciences and Engineering under Award DE-FG02-08ER46544 for his work at the Institute for Quantum Matter (IQM). 
KTW acknowledges support from the National Science Foundation under the Graduate Research Fellowship Program, Fellowship No. DGE-1144245.
This research is part of the Blue Waters sustained-petascale computing project, which is supported by the National Science Foundation (awards OCI-0725070 and ACI-1238993) and the state of Illinois. Blue Waters is a joint effort of the University of Illinois at Urbana-Champaign and its National Center for Supercomputing Applications.

\section*{Author Contributions}
HJC, HZ and LKW conceived the initial DMD ideas and designed the project and organization of the paper. 
All authors contributed to the theoretical developments and various representative \textit{ab initio} and lattice examples. 
All authors contributed to the analysis of the data, discussions and writing of the manuscript. 
LKW oversaw the project. HJC and HZ contributed equally to this work.
 
\section*{Appendix}
In Section~\ref{subsection:3band} we discussed parameterizing the transformation ${\bf T}$ 
as a $4\times12$ matrix for the $2 \times 2$ unit cell, in terms of $\alpha_1$, $\alpha_2$, $\alpha_3$ 
and $\alpha_4$. Using the numbering of the orbitals corresponding to Figure~\ref{fig:threeband}, 
the explicit form of ${\bf T}$ is,
\begin{eqnarray}
{\bf T} = 
\left(
\begin{array}{cccccccccccc}
F        & \alpha_2 &        \alpha_2 &  \alpha_4 & \alpha_1 & \alpha_1 & -\alpha_1 & -\alpha_1 & \alpha_3 & -\alpha_3 & \alpha_3 & -\alpha_3 \\
\alpha_2 &  F       &        \alpha_4 &  \alpha_2 & \alpha_3 & -\alpha_1 & \alpha_1 & -\alpha_3 & -\alpha_3 & \alpha_3 & \alpha_1 & -\alpha_1 \\
\alpha_2 & \alpha_4 & F               &  \alpha_2 & -\alpha_1 & \alpha_3 & -\alpha_3 & \alpha_1 & \alpha_1 & -\alpha_1 & -\alpha_3 & \alpha_3 \\
\alpha_4 & \alpha_2 & \alpha_2        &   F       & -\alpha_3 & -\alpha_3 & \alpha_3 & \alpha_3 & -\alpha_1 & \alpha_1 & -\alpha_1 & \alpha_1 \\
\end{array}
\right)\,,
\end{eqnarray}
where we have defined $F \equiv \sqrt{1-4{\alpha_1}^2 - 2{\alpha_2}^2 - 4 {\alpha_3}^2 -{\alpha_4}^2}$.

A concrete and representative example of our results, shown in Section~\ref{subsection:3band} 
for the $2\times2$ cell, is explained for the case of $U_d/t_{pd}=8$ and $\Delta/t_{pd}=3$. 
The first task was to obtain the optimal transformation ${\bf T}$ 
for which the lowest three eigenstates ($s=1,2,3$) of the three-band model were used for 
computing the cost in Eq.~\eqref{eq:C}. The minimum of the cost was 
determined by a brute force scan in the four dimensional space of $\alpha$'s and 
using a linear grid spacing of $0.002$ found $\alpha_1=0.216$, $\alpha_2=0.042$, $\alpha_3=0.018$ and $\alpha_4=0.016$. 
The two terms in the cost i.e. the trace and orthogonality conditions are individually satisfied to a relative error of 
less than 0.5 percent. 

$\langle c_i^{\dagger} c_j \rangle_s$ and $\langle {c_{j,\uparrow}}^{\dagger} {c_{m,\downarrow}}^{\dagger} c_{n,\downarrow} c_{k,\uparrow} \rangle_{s}$ were computed from the exact knowledge of the three-band model eigenstates and hence $\langle {\tilde{d}_{i,\eta}}^{\dagger} \tilde{d}_{j,\eta} \rangle_{s}$ 
and $\langle \tilde{n}_{i,\uparrow} \tilde{n}_{i,\downarrow} \rangle_{s}$ are obtained 
once the optimal ${\bf T}$ has been determined. 
As mentioned in the main text, the value of $\langle \tilde{d}_{1,\eta}^{\dagger} \tilde{d}_{2,\eta} \rangle_s$ provides estimates for $U/t$ of the effective model by direct comparison of its value to that in the 
corresponding eigenstate in the one-band model. For our test example, the absolute values of $\langle \tilde{d}_{1,\uparrow}^{\dagger} \tilde{d}_{2,\uparrow} \rangle_s$ 
in states $s=1,2,3$ are approximately $0.159$, $0.142$ and $0.084$ respectively which correspond to $(U/t)_1 \approx 14.1 $, $(U/t)_2 \approx 13.2 $, $(U/t)_3 \approx 12.7 $. 
Performing DMD with the three eigenenergies and their calculated RDMs gave $t=0.3025$ eV and $U/t = 13.45 $; the latter in the correct range of the other estimates.

\bibliographystyle{unsrt}
\bibliography{refs}

\end{document}